\documentclass[article]{IEEEtran}
\usepackage{epsfig,latexsym,amsmath,amssymb,setspace,cite,color,graphicx,algorithm,algorithmic,cases}
\usepackage[caption=false,font=footnotesize]{subfig}
\usepackage[percent]{overpic}
\usepackage{physics}

\usepackage{amsmath}
\usepackage{amsfonts}
\usepackage{amssymb}
\usepackage{dsfont}
\usepackage{bm}
\usepackage{amsthm}
\usepackage{newlfont}
\usepackage{float}
\usepackage{hyperref}
\usepackage{algorithm}
\usepackage{algorithmic}
\usepackage{enumerate}
\usepackage{chngcntr}
\usepackage{mathtools}
\usepackage{breqn}
\usepackage{stmaryrd}
\usepackage{cite}
\usepackage{paralist}
\hypersetup{
    colorlinks=true, 
    linkcolor= blue, 
    citecolor=blue 
}
\newcommand\numberthis{\addtocounter{equation}{1}\tag{\theequation}}

\usepackage{acronym}
\acrodef{DMC}[DMC]{Discrete Memoryless Channel}

\DeclareMathOperator*{\argmax}{arg\,max}

\IEEEoverridecommandlockouts

\begin{document}
\sloppy

\title{Multiuser Commitment over Noisy Channels\thanks{R. Chou is with the Department of Computer Science and Engineering  at The University of Texas at Arlington, Arlington, TX. M. Bloch is with the School of Electrical and Computer Engineering at The Georgia Institute of Technology, Atlanta, GA. This work was supported in part by NSF grant CCF-2401373. Part of this
work has been presented at the 58th Annual Allerton Conference on Communication, Control, and Computing \cite{chou2022commitment}. Emails : remi.chou@uta.edu, matthieu.bloch@ece.gatech.edu}}

\author{
  	\IEEEauthorblockN{R\'emi A. Chou and Matthieu R. Bloch
	}

}
\maketitle

\newtheorem{qu}{Question} 
\newtheorem{thm}{Theorem} 
\newtheorem{lem}{Lemma}
\newtheorem{prop}{Proposition}
\newtheorem{cor}{Corollary}
\newtheorem{defn}{Definition}
\newcommand{\remarkend}{\IEEEQEDopen}
\newtheorem{remark}{Remark}
\newtheorem{rem}{Remark}
\newtheorem{pro}{Property}

\newenvironment{example}[1][Example]{\begin{trivlist}
\item[\hskip \labelsep {\bfseries #1}]}{\end{trivlist}}

\theoremstyle{definition}
\newtheorem{ex}{Example}

\renewcommand{\qedsymbol}{ \begin{tiny}$\blacksquare$ \end{tiny} }
\newcommand{\mrb}[1]{\textcolor{black}{#1}}

\renewcommand{\leq}{\leqslant}
\renewcommand{\geq}{\geqslant}

\begin{abstract}
We consider multi-user commitment models that capture the problem of enabling multiple bidders to simultaneously submit auctions to verifiers while ensuring that
    \begin{inparaenum}[i)]
    \item verifiers do not obtain information on the auctions until bidders reveal them at a later stage; and,
    \item bidders cannot change their auction once committed.
    \end{inparaenum}
    Specifically, we assume that bidders and verifiers have access to a noiseless channel as well as a noisy multiple-access channel or broadcast channel, where inputs are controlled by the bidders and outputs are observed by verifiers. In the case of multiple bidders and a single verifier connected by a non-redundant multiple-access channel, we characterize the  commitment capacity region when bidders are not colluding. When the bidders are colluding, we derive an achievable region and a tight converse for the sum rate. In both cases our proposed achievable commitment schemes are constructive. In the case of a single bidder and multiple verifiers connected by a non-redundant broadcast channel, in which verifiers could drop out of the network after auctions are committed, we also characterize the commitment capacity. Our results demonstrate how commitment schemes can benefit from multi-user protocols, and develop resilience when some verifiers may become unavailable. 
  
\end{abstract}

\begin{IEEEkeywords}
Commitment, noisy channels, multiple-access channels, broadcast channel, information-theoretic security.
\end{IEEEkeywords}
\section{Introduction}
Commitment without the need for a trusted third party can be traced back to Blum's coin-flipping problem \cite{blum1983coin}. More generally, a two-party commitment problem involves a bidder, Alice, and a verifier, Bob, and operates in two phases. In the first phase, called the commit phase, Alice sends Bob information  to commit to a message $M$, representing a bid in an auction that must remain concealed from Bob. In the second phase, called the reveal phase, Alice reveals a message $M'$ to Bob, who must determine whether $M'$ is the message that Alice committed to in the commit phase. The protocol must be binding in the sense that, in the reveal phase, Alice cannot make Bob believe that she committed to a message $M'\neq M$. It is well-known that information-theoretic concealment guarantees cannot be achieved over noiseless communication channels and in the absence of additional resources, e.g., \cite{damgaard1999possibility}. However, when a noisy channel is available as a resource, both concealment and binding requirements can be obtained under information-theoretic guarantees, i.e., when Alice and Bob are not assumed to be computationally limited, for some class of noisy channels called non-redundant~\cite{winter2003commitment}.  

Most of the literature on information-theoretic commitments focuses on two-party scenarios that involve a single bidder and a single verifier, e.g., \cite{crepeau1997efficient,damgaard1999possibility,winter2003commitment,damgaard2004unfair,imai2006efficient,oggier2008practical,crepeau2020commitment,nascimento2008commitment,tyagi2015converses,hayashi2022commitment,budkuley2022reverse,yadav2024wiretapped,chou2023retractable}, characterizing commitment capacity under increasingly complex channel models. We study here instead two specific multi-user commitment settings: one in which a verifier interacts with $L$ bidders, each committing to individual messages, and another in which a bidder commits to $B$ verifiers, some of whom may drop out after the commit phase. The motivation for our first setting is to explore whether a multi-bidder protocol can outperform single-bidder protocols and time-sharing when multiple bidders wish to commit to individual messages. Our motivation to consider multiple verifiers in our second setting, is to ensure positive commitment rates even if a verifier drops out of the network after the commit phase, which would be impossible with a single verifier. We note that extensions of the single-user commitment problems can also be found in the computer science literature. For instance, vector commitment~\cite{catalano2013vector} addresses the problem of committing ordered sequences that can later be opened in specific positions. The key contribution therein is to create concise commitments, for which the commitment size is independent of the committed vector length. While the conciseness of the proposed vector commitment bears conceptual similarities with the notion of commitment capacity, the security is computational and does not rely on a noisy channel. Perhaps more closely related to the present work, the notion of commitment with multiple committers and verifiers under information-theoretic security has also been explored in the absence of a noisy channel in \cite{nojoumian2012novel}, when each pair of parties has access to a secure private channel and the majority of parties are honest. Because~\cite{nojoumian2012novel} does not exploit noisy channels, the coding mechanisms enabling commitment are fundamentally different.

In our first setting, $L$ bidders and the verifier have access to a noiseless public communication channel and a noisy discrete memoryless multiple-access channel with $L$ inputs. Each input of the multiple-access channel is controlled by a distinct bidder and the verifier observes the output of the channel. Similar to a two-party setting, the protocol consists of a commit phase and a reveal phase. Here, the concealment requirement is that the verifier must not learn, in an information-theoretic sense, information about any message of any bidder
after the commit phase. The protocol must also be information-theoretically binding in the sense that, during the reveal phase, a bidder cannot make the verifier believe that it committed to another message than the one committed to in the commit phase. For this setting, we consider both the cases of colluding and non-colluding bidders. The non-colluding bidders case corresponds to a scenario in which the bidders do not trust each other and do not want to exchange information with one another. For instance, this would be the case when the bidders commit to messages sent to an auctioneer. Note that the case of colluding bidders may not be suitable for an auction, as bidders are allowed to share their committed sequences. The motivation behind our colluding bidders setting is to enable bulk commitment, for example, in scenarios where each bidder represents a data source, and all data sources belong to the same entity. In such cases, data sources commit their values to the verifier collectively. Example \ref{ex1} illustrates that, for some channels, jointly committing across data sources can be more advantageous than separately committing. Under a non-redundancy condition on the multiple-access channel, we derive the capacity region for the non-colluding bidders case, and an achievable region and the sum-rate capacity for the colluding bidders case. In both cases, our achievability scheme is constructive and relies on distributing hashing with two-universal hash functions~\cite{carter1979universal} for the concealment guarantees. The bindingness of our achievability scheme hinges on the non-redundancy property of the multiple access channel akin to the two-party commitment in~\cite{winter2003commitment}.
The characterization of the sum-rate capacity relies on the polymatroidal properties of our achievability region. Additionally, we demonstrate through a numerical example that, for some channels, using a multi-bidder protocol can outperform using single-bidder protocols, e.g., \cite{winter2003commitment}, and time-sharing.

In our second setting, a single bidder can interact with $B$ verifiers through a noiseless public communication channel and a noisy discrete memoryless broadcast channels with $B$ outputs. The input of the channel is controlled by the bidder and each output is observed by a verifier. Similar to our first setting, the protocol consists of a commit phase and a reveal phase, and must guarantee information-theoretic bindingness and concealment. Our results demonstrate that introducing multiple verifiers can mitigate situations in which verifiers could drop out of the network after the commit phase, 
 as long as one verifier is available during the reveal phase to validate the committed message. For this setting, we derive the commitment capacity under a non-redundancy property of the channel.

We emphasize that the analysis of multi-user commitment goes beyond the single-user case on several accounts:
  \begin{inparaenum}[i)]
  \item the converse proof requires careful considerations to properly handle Markov chains;
  \item the key generation process involves a multi-user version of the hash lemma; and,
  \item we invoke sub-modularity arguments in some of our achievability proofs. 
  \end{inparaenum}

As a byproduct of independent interest, we show  a simple sufficient condition, in terms of injectivity of the transition probability matrix of the channel, to guarantee channel non-redundancy. 
 Additionally, we prove that, for channels whose input alphabet is at most three, our sufficient condition is also necessary, and thus equivalent to the non-redundancy  characterization  in \cite{winter2003commitment}.

The remainder of the paper is organized as follows. After a review of notation in Section~\ref{sec:notation}, we develop  an auxiliary result in Section~\ref{sec:non-redundant}, offering a sufficient condition to guarantee non-redundancy and an alternative characterization of non-redundant channels with input alphabet  of cardinality at most three. We formally introduce our multiuser commitment models in Section~\ref{secmod} along with the associated results. We delegate the proofs to Section~\ref{sec:proof-theo-mac} and Section~\ref{sec:proof-theorem-single-bid-multi-ver}. Finally, we provide concluding remark in~Section~\ref{seccl}.

\section{Notation}
\label{sec:notation}

 For $a,b \in \mathbb{R}$, define $\llbracket a,b \rrbracket \triangleq [\lfloor a \rfloor , \lceil b \rceil ] \cap \mathbb{N}$. Unless specified otherwise, random variables and their realizations are denoted by uppercase and corresponding lowercase letters, respectively, e.g., $x$ is a realization of the random variable $X$.  A \ac{DMC} with input alphabet $\mathcal{X}$, output alphabet $\mathcal{Y}$, and transition probability $W$ is denoted by $(\mathcal{X},\mathcal{Y},W)$.  Additionally, for $x \in \mathcal{X}$, $W_x$ denotes the output distribution of the channel when the input is $x$. The probability simplex for distributions defined over the set $\mathcal{X}$ is denoted by $\mathcal{P}(\mathcal{X})$. For a \ac{DMC} $(\mathcal{X},\mathcal{Y},W)$ and a distribution $p \in \mathcal{P}(\mathcal{X})$, $W\circ p$ denotes the  distribution of the channel output when the input is distributed according to $p$. $\Vert \cdot \Vert$ denotes the $\ell_1$-norm. For two distributions $p,q \in \mathcal{P}(\mathcal{X})$, the variational distance between $p$ and $q$ is denoted by $\mathbb{V}(p,q) \triangleq \frac{1}{2}\Vert p-q \Vert $. The indicator function is denoted by $\mathds{1}\{ \omega \}$, which is equal to~$1$ if the predicate $\omega$ is true and $0$ otherwise.
\section{Non-Redundant Channels}
\label{sec:non-redundant}
The concept of non-redundant channel was introduced in~\cite{winter2003commitment}. A \ac{DMC} $(\mathcal{X},\mathcal{Y},W)$ is \emph{non-redundant}~if 
\begin{align} 
\exists \eta>0,\forall x \in \mathcal{X}, \forall p \in \mathcal{P}(\mathcal{X}) \text{ s.t. }   p(x) =0, \nonumber\\
\Vert W_{x} - W\circ p \Vert \geq \eta.\label{eq:non-redundant}
\end{align}
The following proposition offers a sufficient condition for a channel to be non-redundant. We will then show that this condition is also necessary for a channel to be non-redundant in the case of an input alphabet with cardinality smaller than or equal to three.

\begin{prop}\label{prop1}
    A \ac{DMC}   $(\mathcal{X},\mathcal{Y},W)$ is non-redundant if  $$W\circ p=W\circ q\Rightarrow p=q.$$
\end{prop}
\begin{proof}
   We first show that a channel is non-redundant if and only if\footnote{This fact is implicitly assumed true in~\cite{winter2003commitment}, we prove it here for completeness.}
  \begin{align}
    \forall x \in \mathcal{X}, \forall p \in \mathcal{P}(\mathcal{X}) \text{ s.t. }   p(x) =0,  W_{x} \neq W\circ p.\label{eq:non-redundant-2}
  \end{align}  
  The fact that~(\ref{eq:non-redundant}) implies~(\ref{eq:non-redundant-2}) follows from the definition. Assume now that~(\ref{eq:non-redundant-2}) holds. For $x\in\mathcal{X}$, define the functional
  \begin{align*}
    f_x:\Delta(\mathcal{X}\setminus\{x\})&\to \mathbb{R}_+\\ p& \mapsto \Vert W_{x} -  W\circ p\Vert,
  \end{align*}
  where $\Delta(\mathcal{X}\setminus\{x\})$ is the set of distributions in $\mathcal{P}(\mathcal{X})$ with supports included in $\mathcal{X}\setminus\{x\}$.
    For any $p_1,p_2\in\Delta(\mathcal{X}\setminus\{x\})$, the triangle inequality ensures that
  \begin{align*}
    \abs{f(p_1)-f(p_2)}& = \abs{\Vert W_{x} -  W\circ p_1\Vert-\Vert W_{x} -  W\circ p_2\Vert}\\
    &\leq \Vert{W\circ p_1-W\circ p_2}\Vert\\ &\leq \Vert p_1-p_2\Vert,
  \end{align*}
so that $f_x$ is continuous on $\Delta(\mathcal{X}\setminus\{x\})$. Since $\Delta(\mathcal{X}\setminus\{x\})$ is a subset of $\mathbb{R}^{\abs{\mathcal{X}}-1}$ that is closed and bounded, it is compact and there exists $p_x\in \Delta(\mathcal{X}\setminus\{x\})$ such that
\begin{align*} \inf_{p\in\Delta(\mathcal{X}\setminus\{x\})}f_x(p) = f_x(p_x).
\end{align*}
Setting $\eta = \min_{x\in\mathcal{X}}f_x(p_x)$ shows that~(\ref{eq:non-redundant}) holds, so that~(\ref{eq:non-redundant}) and~(\ref{eq:non-redundant-2}) are indeed equivalent.

Finally, if the channel is such that $W\circ p=W\circ q\Rightarrow p=q$, then~(\ref{eq:non-redundant-2}) follows directly and the channel is non-redundant.
\end{proof}

\begin{prop}
  \label{prop:non-redundant}
  If $\abs{\mathcal{X}}\leq 3$, a \ac{DMC}   $(\mathcal{X},\mathcal{Y},W)$ is non-redundant if and only if $$W\circ p=W\circ q\Rightarrow p=q.$$
\end{prop}

\begin{proof}
By Proposition \ref{prop1}, it only remains to show that  $W\circ p=W\circ q\Rightarrow p=q$ is a necessary condition for non-redundancy. By contraposition, assume now that the channel is such that there exist two distinct distributions $p,q\in\mathcal{P}(\mathcal{X})$ such that $W\circ p=W\circ q$. Define $\mathcal{S}\triangleq \{x\in\mathcal{X}:p(x)>q(x)\}$, which is not empty since $p$ and $q$ are distinct so that $|\mathcal{S}|\geq 1$. Define the distributions
  \begin{align*}
    \tilde{p}(x) & \triangleq \frac{p(x)-q(x)}{\sum_{u\in\mathcal{S}} (p(u)-q(u))}, \forall x\in\mathcal{S},\\
    \tilde{q}(x) & \triangleq \frac{q(x)-p(x)}{\sum_{u\in\mathcal{S}^c} (q(u)-p(u))}, \forall x\in\mathcal{S}^c,
  \end{align*}
  which have disjoint support, and note that $$\sum_{u\in\mathcal{S}} (p(u)-q(u))=\sum_{u\in\mathcal{S}^c} (q(u)-p(u)).$$ Then, we have
  \begin{align*} 
    &W\circ p=W\circ q\\ &\Leftrightarrow \forall y\in\mathcal{Y}, \textstyle\sum_{x\in\mathcal{S}}W(y|x)(p(x)-q(x)) \\
    & \phantom{-------}= \textstyle\sum_{x\in\mathcal{S}^c}W(y|x)(q(x)-p(x))\\
                      &\Leftrightarrow \forall y\in\mathcal{Y}, \textstyle\sum_{x\in\mathcal{S}}W(y|x)\tilde{p}(x) = \textstyle\sum_{x\in\mathcal{S}^c} W(y|x)\tilde{q}(x).
  \end{align*}
  Since $\abs{\mathcal{X}}\leq 3$, either $|\mathcal{S}|=1$ or $|\mathcal{S}^c|=1$ and we can assume without loss of generality that $\mathcal{S}=\{x^*\}$ and $\tilde{p}(x^*)=1$. Hence, $W_{x^*}=\sum_{x\in\mathcal{X}\setminus\{x^*\}}W(y|x)\tilde{q}(x)$ and the channel is redundant as per~(\ref{eq:non-redundant-2}).
\end{proof}

Proposition~\ref{prop:non-redundant} provides an alternative characterization of the non-redundant condition but the result cannot be extended to the case $|\mathcal{X}|>3$. Indeed, the channel with $\abs{\mathcal{X}}=4$, $\abs{\mathcal{Y}}=3$, and 
\begin{align}
W=  \left[
  \begin{array}{ccc}
    1&	0& 	0\\
    1/2&	 1/2&	0\\
     1/2&	0&  1/2\\
    0&	 1/2	& 1/2\\
  \end{array}
  \right]
\end{align}
corresponds to a non-redundant channel since the rows define the extreme points of a polytope but the  input distributions $(0.5,0,0,0.5)$ and $(0,0.5,0.5,0)$ induce the same output distribution.

\section{Multi-User Commitment Models and Main Results} \label{secmod}

In Section \ref{sec:multi-bid-single-ver}, we present a multi-bidder single-verifier model and our results for both colluding and non-colluding bidders settings. In Section \ref{sec:single-bidder-multi}, we present a single-bidder multi-verifier model and our results for this model.
\subsection{Multi-Bidder Single-Verifier Model and Results}
\label{sec:multi-bid-single-ver}

\begin{figure}
  \centering
  \includegraphics[width=7cm]{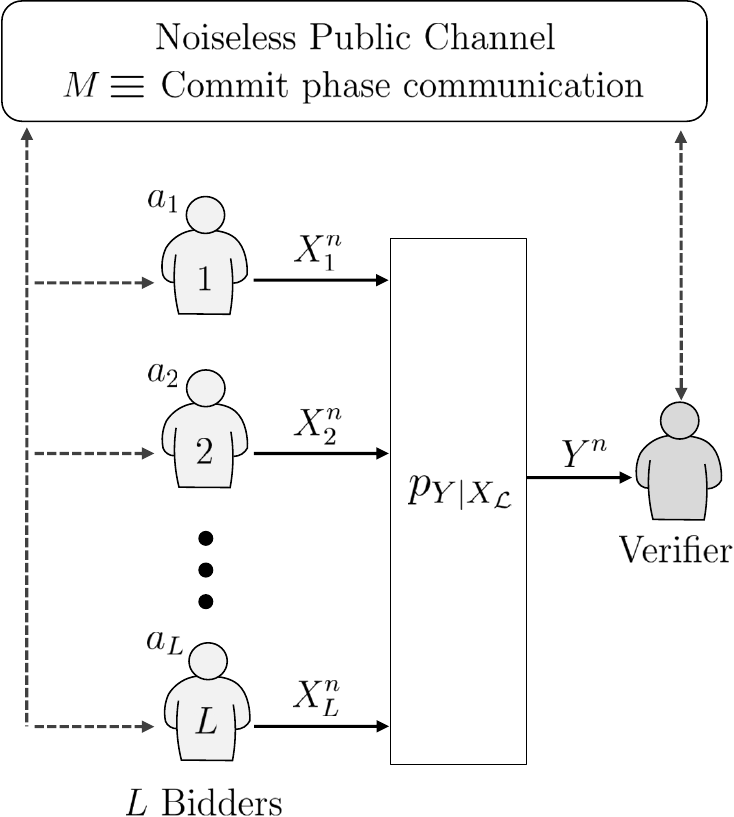}
  \caption{Multi-bidder single-verifier commitment over a multiple-access channel.}
  \label{fig:multi-bid-single-ver}
\end{figure}
We first consider the model illustrated in Figure~\ref{fig:multi-bid-single-ver}, in which $L \in \mathbb{N}^*$ bidders want to commit messages to a single verifier using a noiseless channel and a noisy multiple access channel. We set $\mathcal{L} \triangleq \llbracket 1 , L \rrbracket$. The multiple-access channel is characterized by $L$ finite input alphabets $(\mathcal{X}_{\ell})_{\ell \in \mathcal{L}}$, a finite output alphabet $\mathcal{Y}$, and a transition probability $p_{Y|X_1,\cdots,X_L}$. To simplify notation, we denote the Cartesian product of the input alphabets by $\mathcal{X}_{\mathcal{L}} \triangleq \bigtimes_{{\ell} \in \mathcal{L}} \mathcal{X}_{\ell}$, a generic input by $X_{\mathcal{L}} \triangleq (X_{\ell})_{\ell \in \mathcal{L}}$ where $X_{\ell}$, ${\ell} \in\mathcal{L}$, is defined over $\mathcal{X}_{\ell}$, and $p_{Y|X_{\mathcal{L}}}\triangleq p_{Y|X_1,\cdots,X_L}$. For any $x_{\mathcal{L}} \in \mathcal{X}_{\mathcal{L}}$, define  $W_{x_{\mathcal{L}}} : y \mapsto p_{Y|X_{\mathcal{L}}}(y|x_{\mathcal{L}})$. 

We assume throughout that the multiple access channel $W$ is non-redundant by adapting the definition in~(\ref{eq:non-redundant-2}), i.e.,
\begin{align} 
\forall x_{\mathcal{L}} \in \mathcal{X}_{\mathcal{L}}, \forall p_{X_{\mathcal{L}}} \in \mathcal{P}(\mathcal{X}_{\mathcal{L}}) \text{ s.t. }   p_{X_{\mathcal{L}}}(x_{\mathcal{L}}) =0, \nonumber \\ W_{x_{\mathcal{L}}}  \neq \sum_{x'_{\mathcal{L}}\in \mathcal{X}_{\mathcal{L}}} p_{X_{\mathcal{L}}}(x'_{\mathcal{L}})W_{x'_{\mathcal{L}}}.\label{eqred}
\end{align}

The bidders interactively use the noiseless and noisy channels to commit their messages (commit phase) and reveal their messages (reveal phase). We distinguish two modes of operation, depending on whether the bidders collude or not.

\subsubsection{Colluding bidders}
\label{sec:coll-bidders}

For $({R_{\ell}})_{{\ell} \in\mathcal{L}}$, a $((2^{n{R_{\ell}}})_{{\ell} \in\mathcal{L}},n)$ commitment scheme with colluding bidders consists~of:
\begin{itemize}
\item For every ${\ell} \in\mathcal{L}$, a sequence $a_{\ell} \in \mathcal{A}_{\ell}  \triangleq \llbracket 1 , 2^{nR_{\ell}} \rrbracket$ that Bidder $\ell$ wishes to commit to;
\item A noiseless communication channel between the bidders and the verifier;
\item Local randomness $S \in \mathcal{R}$ available at the bidders;
\item Local randomness $S'\in \mathcal{R}'$ available at the verifier.
\end{itemize}
The scheme then operates as follows:
\begin{enumerate}[(i)]
    \item Commit phase: Define $a_{\mathcal{L}} \triangleq (a_{\ell})_{{\ell} \in \mathcal{L}}$. For channel use $i\in \llbracket 1 , n \rrbracket$, the bidders send $X_{\mathcal{L},i}(a_{\mathcal{L}},S,M'_{1:i-1,1:r_{i-1}})$ over the channel, where  $M'_{1:i-1,1:r_{i-1}} \triangleq  (M'_{\alpha,\beta} )_{\alpha \in \llbracket 1, i-1 \rrbracket, \beta \in \llbracket 1, r_{i-1}\rrbracket} $ (with the convention $M'_{1:i-1,1:r_{i-1}} = \emptyset $ when $i=0$) represents messages previously received from the verifier as described next. Then, the bidders engage in $r_i$ rounds of noiseless  communication with the verifier, i.e., for $j \in \llbracket 1 , r_i \rrbracket$, the bidders send  $M_{i,j} (a_{\mathcal{L}},S,M'_{1:i,1:j-1})$ and the verifier replies $M'_{i,j} (S',M_{1:i,1:j},Y^i)$, where $Y^i \triangleq (Y_1, Y_2, \ldots, Y_i)$ and $M_{1:i,1:j} \triangleq  (M_{\alpha,\beta} )_{\alpha \in \llbracket 1, i \rrbracket, \beta \in \llbracket 1, j\rrbracket}$. We denote the collective noiseless communication between the bidders and the verifier by $M$, i.e., $M \triangleq (M'_{1:n,1:r_{n}},M_{1:n,1:
    r_{n}})$. Define $V\triangleq (Y^n, M, S')$ as all the information available at the verifier at the end of the commit phase.  
\item Reveal phase: The bidders reveal $(a_{\mathcal{L}},S)$ to the verifier. The verifier performs a test $\beta(V, a_{\mathcal{L}}, S ) $ that returns $1$ if the sequence $a_{\mathcal{L}}$ is accepted and $0$ otherwise.
\end{enumerate}
The collusion of the bidders is reflected in the randomness $S$ common to \emph{all} bidders and the fact that bidders know each other's messages in the commit phase.

\begin{defn} \label{defreqs}
A rate-tuple $(R_{\ell})_{\ell \in \mathcal{L}}$ is achievable if there exists a sequence of $((2^{n{R_{\ell}}})_{{\ell} \in\mathcal{L}},n)$ commitment schemes such that for any $\tilde S \in \mathcal{R}$, $ a_{\mathcal{L}},\tilde{a}_{\mathcal{L}} \in \mathcal{A}_{\mathcal{L}}$ such that $a_{\mathcal{L}} \neq \tilde{a}_{\mathcal{L}}$,
\begin{align}
 \lim_{n\to \infty} \mathbb{P}[ \beta(V, a_{\mathcal{L}}, S) =  0] &= 0,  \text{ (correctness) }\label{eqcorrectness}\\
 \lim_{n\to \infty} I(A_{\mathcal{L}}; V) &= 0, \text{ (concealment)}\label{eqconcealment}\\
  \lim_{n\to \infty}\mathbb{P}[\beta(V, a_{\mathcal{L}}, S)\! =\!  1 \! =\! \beta(V, \tilde{a}_{\mathcal{L}}, \tilde S)]    &= 0. \text{ (bindingness)} \label{eqbindingness}
\end{align}
The set of all achievable rate-tuples is the capacity~region, and the supremum, over all achievable rate-tuples $(R_{\ell})_{{\ell} \in\mathcal{L}}$, of the sum-rates $\sum_{{\ell} \in\mathcal{L}}R_{\ell}$ is the sum-rate capacity.
\end{defn}
Equation \eqref{eqcorrectness} means that the verifier accepts the revealed sequences $a_{\mathcal{L}}$ as long as the bidders are honest. Equation~\eqref{eqconcealment} ensures that the verifier gains no information   about the committed sequences $a_{\mathcal{L}}$ prior to the reveal phase. Equation~\eqref{eqbindingness} ensures that the bidders cannot convince the verifier to accept sequences other than those initially committed to.

Our main result is a partial characterization of the capacity region for  colluding bidders. In the following, for any $\mathcal{T} \subseteq \mathcal{L}$, we write the sum-rate of the bidders in $\mathcal{T}$ as $R_{\mathcal{T}} \triangleq \sum_{{\ell} \in \mathcal{T}} R_{\ell}$.

\begin{thm} \label{th1}
For the case of colluding bidders, the following region is achievable:
$$\bigcup_{p_{X_{\mathcal{L}}}\in {\mathcal{P}}(\mathcal{X}_{\mathcal{L}})} \{ (R_{\ell})_{{\ell} \in\mathcal{L}} : R_{\mathcal{T}}  \leq  H(X_{\mathcal{T}}|Y) , \forall \mathcal{T} \subseteq \mathcal{L}\}.   $$
Moreover, the sum-rate capacity is $\max_{p_{X_{\mathcal{L}} \in {\mathcal{P}}(\mathcal{X}_{\mathcal{L}})}  } H(X_{\mathcal{L}}|Y).$
\end{thm}
\begin{proof}
  See Section~\ref{secach}.
\end{proof}

The following example demonstrates that, for some channels, having a multi-bidder coding scheme, rather than doing time-sharing with single-bidder coding schemes~\cite{winter2003commitment},  can increase the sum-rate of the bidders. 
\begin{ex}\label{ex1}
Suppose $L=2$ and $\mathcal{X}_1=\mathcal{X}_2=\mathcal{Y}=\{0,1\}$. Consider the following transition probability matrices
    \begin{align}
W&=  \left[
  \begin{array}{ccc}
    W(0|0,0)&	W(1|0,0)\\
    W(0|0,1)&	W(1|0,1)\\
    W(0|1,0)&	W(1|1,0)\\
    W(0|1,1)&	W(1|1,1)	\\
  \end{array}
  \right], \label{channelW}\\
W_1&=  \left[
  \begin{array}{ccc}
    W(0|0,0)&	W(1|0,0)\\
    W(0|0,1)&	W(1|0,1)
  \end{array}
  \right], \label{channelW1}\\
  W_2 &=  \left[
  \begin{array}{ccc}
    W(0|0,0)&	W(1|0,0)\\
    W(0|1,0)&	W(1|1,0)  \end{array}
  \right],\label{channelW2}
\end{align}
and suppose that the corresponding channels are non-redundant.

The transition probability matrix $W_i$, $i\in\{1,2\}$, corresponds to the situation in which only Bidder $(3-i)$ uses the channel $W$ and $X_i$ is constant and equal to $0$, i.e., when the distribution of $X_i$ is $p_{X_i}^{\star}$ defined by $(p_{X_i}^{\star}(0),p_{X_i}^{\star}(1))\triangleq (1,0)$. For the problem of commitment as described in Definition~\ref{defreqs}, let $\Sigma(W)$ be the sum-rate capacity for the channel $W$ and $C(W_i)$, $i\in\{1,2\}$, be the commitment capacity for the channel $W_i$. For convenience, we write $H_{p}(X_1X_2|Y)$ to indicate that $(X_1,X_2)$ is distributed according to $p$. Then, from Theorem~\ref{th1}, we have
\begin{align*}
\Sigma(W)
&=\max_{p_{X_1X_2 }  } H(X_1X_2|Y) \\
&\geq H_{p^{\star}}(X_1X_2|Y) \\
&= H_{p^{\star}}(X_1|Y)+H_{p^{\star}}(X_2|YX_1) \\
&= C(W_1)+H_{p^{\star}}(X_2|YX_1)\\
& \geq C(W_1), \numberthis  \label{eqlbSW}
\end{align*}
where the first inequality holds with $p^{\star}  \in \argmax_{p_{X_1}p_{X_2}^{\star}} H(X_1|Y) $. Similar to \eqref{eqlbSW}, we have $\Sigma(W)\geq C(W_2)$ such that 
\begin{align}
\Sigma(W)\geq \max(C(W_1),C(W_2)).\label{eqlbS}
\end{align}
Equation \eqref{eqlbS} indicates that using time-sharing with a single bidder over $W_1$ and $W_2$ cannot outperform  a  multi-bidder coding scheme over $W$.

We now provide a numerical example to show that the inequality in \eqref{eqlbS} can be strict, demonstrating that
 using $W$ with a multi-bidder coding scheme achieves a larger sum rate than using time-sharing with a single bidder over $W_1$ and $W_2$.  Specifically, consider the following channels
 \begin{align}
W&=    \left[
  \begin{array}{ccc}
    1/4&	3/4\\
    1/2&	1/2\\
    1/2&	1/2\\
    1/2&	1/2	\\
  \end{array}
  \right],\label{exW}\\
W_1&=    \left[
  \begin{array}{ccc}
    1/4&	3/4\\
    1/2&	1/2
  \end{array}
  \right], \label{exW1}\\
  W_2 &=  \left[
  \begin{array}{ccc}
    1/4&	3/4\\
    1/2&	1/2\\
  \end{array}
  \right]. \label{exW2}
\end{align}
Note that the kernel of the matrices $W$, $W_1$, and $W_2$ contain only the zero vector, hence, by Proposition \ref{prop1}, the three channels are non-redundant. Then, from Theorem \ref{th1}, we numerically find~that 
\begin{align*}
\Sigma(W) & > 1.9647 \textnormal{ bits/channel use} 
\\& > 0.9512 \textnormal{ bits/channel use}\\
& > \max( C(W_1),C(W_2)). \end{align*}
\end{ex}

\subsubsection{Non-colluding bidders}
\label{sec:non-colluding-bidders}
For $({R_{\ell}})_{{\ell} \in\mathcal{L}}$, a $((2^{n{R_{\ell}}})_{{\ell} \in\mathcal{L}},n)$ commitment scheme for non-colluding bidders consists~of:
\begin{itemize}
\item For every ${\ell} \in\mathcal{L}$, a sequence $a_{\ell} \in \mathcal{A}_{\ell}  \triangleq \llbracket 1 , 2^{nR_{\ell}} \rrbracket$ that Bidder $\ell$ wants to commit to;
\item Local randomness $S_{\ell} \in \mathcal{R}_{\ell}$ at Bidder $\ell \in \mathcal{L}$;
\item Local randomness $S'_{\ell} \in \mathcal{R}_{\ell}'$, $\ell \in \mathcal{L}$, at the verifier where $S'_{\ell}$ is only used in the interactive noiseless communication with Bidder $\ell \in \mathcal{L}$ during the commit phase.
\end{itemize}
The scheme then operates as follows:
\begin{enumerate}[(i)]
\item Commit phase: We adopt the same notation as in Section \ref{sec:coll-bidders}. For each $\ell \in \mathcal{L}$, for $i\in \llbracket 1 , n \rrbracket$, Bidder~$\ell$ sends $X_{l,i}(a_{\ell},S_{\ell},M'_{l,1:i-1,1:r_{i-1}})$ over the channel and engage in $r_i$ rounds of noiseless  communication with the verifier, i.e., for $j \in \llbracket 1 , r_i \rrbracket$, Bidder~$\ell$ sends  $M_{l,i,j} (a_{\ell},S_{\ell},M'_{l,1:i,1:j-1})$ and the verifier replies $M'_{l,i,j} (S'_{\ell},M_{l,1:i,1:j},Y^i)$. We denote the collective noiseless communication between Bidder $\ell \in \mathcal{L}$ and the verifier by $M_{\ell}$ and define $M\triangleq (M_{\ell})_{{\ell} \in\mathcal{L}}$. Define $V_{\ell} \triangleq (Y^n, M_{l}, S'_{\ell})$ for $\ell \in \mathcal{L}$ and $V_{\mathcal{L}} \triangleq (V_{l})_{{\ell} \in\mathcal{L}}$.
\item Reveal phase: Bidder $\ell \in \mathcal{L}$ reveals $(a_{\ell},S_{\ell})$ to the verifier. For each ${\ell} \in\mathcal{L}$, the verifier performs a test $\beta_{\ell}(V_{\ell}, a_{\ell}, S_{\ell} ) $, ${\ell} \in \mathcal{L}$, that returns $1$ if the sequence $a_{\ell}$ is accepted and $0$ otherwise.
\end{enumerate}
Unlike the case of colluding bidders, the bidders here only have access to their own message and  local randomness. Additionally, in the reveal phase, the acceptance test performed by the verifier is now done individually for each bidder rather than for all the bidders jointly.

\begin{defn}\label{def2}
A rate-tuple $(R_{\ell})_{\ell \in \mathcal{L}}$ is achievable if there exists a sequence of $((2^{n{R_{\ell}}})_{{\ell} \in\mathcal{L}},n)$ commitment schemes such that for any ${\ell} \in\mathcal{L}$, $\tilde S_{\ell}\in \mathcal{R}_{\ell}$, $ a_{\mathcal{L}},\tilde{a}_{\mathcal{L}} \in \mathcal{A}_{\mathcal{L}}$ such that $a_{\mathcal{L}} \neq \tilde{a}_{\mathcal{L}}$,
\begin{align*}
  \lim_{n\to \infty} \mathbb{P}[ \beta_{\ell}(V_{\ell}, a_{\ell}, S_{\ell}) =  0] &= 0,  \text{ (correctness) } \\
 \lim_{n\to \infty} I(A_{\mathcal{L}}; V_{\mathcal{L}}) &= 0, \text{ (concealment)}\\
  \lim_{n\to \infty}\mathbb{P}[\beta_{\ell}(V_{\ell}, a_{\ell}, S_{\ell})\! =\!  1 \! =\! \beta_{\ell}(V_{\ell}, \tilde{a}_{\ell}, \tilde S_{\ell})]    &= 0. \text{ (bindingness)}
\end{align*}
The set of all achievable rate-tuples is the capacity~region, and the supremum, over all achievable rate-tuples $(R_{\ell})_{{\ell} \in\mathcal{L}}$, of the sum-rates $\sum_{{\ell} \in\mathcal{L}}R_{\ell}$ is the sum-rate capacity.
\end{defn}
The correctness, concealment, and bindingness requirements in Definition \ref{def2} can be interpreted similarly to those in the colluding bidders setting discussed in Section \ref{sec:coll-bidders}. Note that since $V_{\ell}$, ${\ell} \in\mathcal{L}$, depends on $Y^n$, which depends on all bidder inputs to the channel, $I(A_{\mathcal{L}}; V_{\mathcal{L}})$ is not necessarily equal to $ \sum_{{\ell} \in\mathcal{L}} I(A_{\ell}; V_{\ell})$, and therefore the concealment requirement from Definition~\ref{defreqs} cannot be simplified. Our main result is a complete characterization of the capacity region for non-colluding bidders. Again, for any $\mathcal{T} \subseteq \mathcal{L}$, we set $R_{\mathcal{T}} \triangleq \sum_{{\ell} \in \mathcal{T}} R_{\ell}$.
 
\begin{thm} \label{th2}
Define the set of product input distributions
\begin{align*}
\mathcal{P}^{\textup{I}}(\mathcal{X}_{\mathcal{L}}) \triangleq \{ p_{X_{\mathcal{L}}} \in \mathcal{P} (\mathcal{X}_{\mathcal{L}}) : p_{X_{\mathcal{L}}} = \textstyle\prod_{{\ell} \in \mathcal{L}}p_{X_{\ell}}\}.
\end{align*}
For the case of non-colluding bidders, the capacity region is  
$$\bigcup_{p_{X_{\mathcal{L}}}\in {\mathcal{P}^{\textup{I}}}(\mathcal{X}_{\mathcal{L}})} \{ (R_{\ell})_{{\ell} \in\mathcal{L}} : R_{\mathcal{T}}  \leq  H(X_{\mathcal{T}}|Y) , \forall \mathcal{T} \subseteq \mathcal{L}\}.   $$
Moreover, the sum-rate capacity is $\max_{p_{X_{\mathcal{L}} \in {\mathcal{P}^{\textup{I}}}(\mathcal{X}_{\mathcal{L}})}  } H(X_{\mathcal{L}}|Y).$
\end{thm}
\begin{proof}
  See Section~\ref{secach}.
\end{proof}

Similar to the scenario with colluding bidders, the following example illustrates that, for some channels, a multi-bidder coding scheme can again  outperform a time-sharing approach that utilizes single-bidder coding schemes~\cite{winter2003commitment}.
\begin{ex}\label{ex3}
Suppose $L=2$. Consider the non-redundant channels $W$, $W_1$, and $W_2$ in \eqref{channelW}, \eqref{channelW1}, \eqref{channelW2}. For the problem of commitment as described in Definition \ref{def2}, let $\Sigma^{\textup{I}}(W)$ be the sum-rate capacity for the channel $W$ and $C(W_i)$, $i\in\{1,2\}$, be the commitment capacity for the channel $W_i$. Using Theorem \ref{th2}, similar to \eqref{eqlbS}, we have
\begin{align}
\Sigma^{\textup{I}}(W)\geq \max(C(W_1),C(W_2)). \label{lbSI}
\end{align}
Next, we numerically show that the inequality \eqref{lbSI} can be strict
such that using $W$ with a multi-bidder coding scheme achieves a larger sum rate than using time-sharing with a single bidder over $W_1$ and $W_2$. Specifically, for the non-redundant channels $W$, $W_1$, and $W_2$ in \eqref{exW}, \eqref{exW1}, \eqref{exW2},  we have
\begin{align*}
\Sigma^{\textup{I}}(W) &> 1.9645 \textnormal{ bits/channel use}\\
&> 0.9512 \textnormal{ bits/channel use} \\
&> \max(C(W_1),C(W_2)). \end{align*}
\end{ex}

\subsection{Single-Bidder Multi-Verifier Model and Results}
\label{sec:single-bidder-multi}

\begin{figure}
  \centering
    \includegraphics[width=7cm]{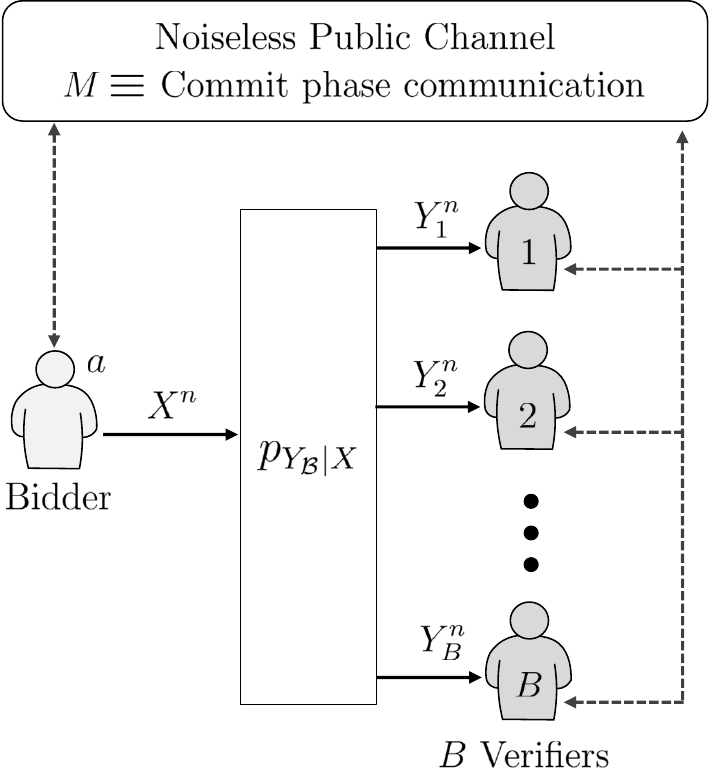}
  \caption{Single-bidder multi-verifier commitment over a broadcast channel.}
  \label{fig:single-bid-multi-ver}
\end{figure}

We now consider the model illustrated in Figure~\ref{fig:single-bid-multi-ver}, in which a single bidder attempts to commit a message to $B\in \mathbb{N}^*$ verifiers using a noiseless channel and a noisy broadcast channel. The verifiers are indexed in the set  $\mathcal{B} \triangleq \llbracket 1 , B \rrbracket$ and are assumed non-colluding. The broadcast channel is characterized by a finite  input alphabet  $\mathcal{X}$, $B$ finite output alphabets $(\mathcal{Y}_b)_{b \in \mathcal{B}}$, and a transition probability $p_{Y_{\mathcal{B}}|X }$, where we use the same notation as in Section \ref{sec:multi-bid-single-ver}. 
For any $b\in \mathcal{B}$, we also consider the channel $W^{(b)} \triangleq ( \mathcal{X},\mathcal{Y}_{b}, p_{Y_{b}|X })$.  
Throughout this section, we assume that  $W^{(b)}$ is non-redundant for any $b\in \mathcal{B}$.

For $R\in\mathbb{R}_+$, a $(2^{n{R}},n)$ commitment scheme consists~of:
\begin{itemize}
\item A sequence $a \in \mathcal{A} \triangleq \llbracket 1 , 2^{nR} \rrbracket$ that the bidder wants to commit to;
\item A noiseless public channel between the bidder and the verifiers; 
\item Local randomness $S \in \mathcal{R}$ at the bidder;
\item Local randomness $S'_b \in \mathcal{R}_b'$, $b \in \mathcal{B}$, at Verifier $b$.
\end{itemize}
The scheme then operates in two phases as follows:
\begin{enumerate}
\item Commit phase: We use the same notation as in Section \ref{sec:multi-bid-single-ver}. For $i\in \llbracket 1 , n \rrbracket$, the bidder sends $X_{i}(a,S,(M'_{b,1:i-1,1:r_{i-1}})_{b\in\mathcal{B}})$ over the channel and engage in $r_i$ rounds of noiseless communication with the verifiers, i.e., for $j \in \llbracket 1 , r_i \rrbracket$ and $b\in\mathcal{B}$, the bidder  sends  $M_{b,i,j} (a,S,M'_{b,1:i,1:j-1})$ and Verifier $b$ replies $M'_{b,i,j} (S'_b,M_{b,1:i,1:j},Y_b^i)$.

We denote the collective noiseless  communication between the bidder and Verifier $b\in \mathcal{B}$ by $M_b$ and define $M\triangleq (M_b)_{b\in\mathcal{B}}$. For $b\in  \mathcal{B}$, define $V_{b} \triangleq (Y_{b}^n, M, S'_{b})$.
 \item Reveal phase: Suppose that some verifiers may have dropped out of the network after the commit phase and  $\mathcal{A} \subset \mathcal{B}$ is the set of available verifiers with $|\mathcal{A}|\geq 1$. The bidder reveals $(a,S)$ and chooses a verifier $b^{\star} \in \mathcal{A}$ at random.   Verifier $b^{\star}$ performs a test $\beta(V_{b^{\star}}, a, S ) $ that returns $1$ if the sequence $a$ is accepted and $0$ otherwise.
\end{enumerate}

\begin{defn} \label{def3}
A rate  $R$ is achievable if there exists a sequence of $(2^{n{R}},n)$ commitment schemes such that for any  $\tilde S\in \mathcal{R}$, $ a,\tilde{a} \in \mathcal{A}$ such that $a \neq \tilde{a}$,
\begin{align}
  \lim_{n\to \infty} \max_{b\in \mathcal{B}} \mathbb{P}[ \beta(V_{b}, a , S ) =  0] &= 0,  \text{ (correctness) }  \\
 \lim_{n\to \infty} \max_{b\in \mathcal{B}} I(A; V_{b}) &= 0, \text{ (concealement) } \label{eqconcealbr}\\
  \lim_{n\to \infty}\max_{b\in \mathcal{B}} \mathbb{P}[\beta(V_{b }, a , S )\! =\!  1 \! =\! \beta (V_{b}, \tilde{a} , \tilde S )]    &= 0.  \text{ (bindingness) }
\end{align}
The supremum of all achievable rates is the commitment capacity.
\end{defn}
The correctness, concealment, and bindingness requirements in Definition \ref{def3} can be interpreted similarly to those in the  settings discussed in Section \ref{sec:multi-bid-single-ver}.  Note that Definition \ref{def3} ensures that correctness, concealement, and bindingness are satisfied irrespective of $\mathcal{A}\subset\mathcal{B} $ such that $|\mathcal{A}|\geq 1$, and for any choice $b^{\star}\in \mathcal{A}$.

Our main result for the multi-verifier commitment setting is to characterize the commitment capacity as follows.

\begin{thm} \label{thm3}
  The commitment capacity is 
 $$\max_{p_X\in {\mathcal{P}}(\mathcal{X})} \min_{b\in \mathcal{B}} H(X|Y_{b}).$$
\end{thm}

\begin{proof}
    See Section~\ref{sec:proof-theorem-single-bid-multi-ver}.
\end{proof}
Theorem \ref{thm3} shows that that a non-zero commitment capacity can still be achievable, even if all but one verifier drop out of the network after the commit phase.
\begin{ex}
Assume  that  the bidder and the $B$ verifiers are connected through $B$  parallel channels that are non-redundant and identical. By Theorem \ref{thm3}, if all but one verifier drop out after the commit phase, the commitment capacity is $\max_{p_X\in {\mathcal{P}}(\mathcal{X})}  H(X|Y_{1})$, which can be non-zero for some non-redundant channels, e.g., the channel in \eqref{exW1}.
\end{ex}

\section{Proof of Theorem~\ref{th1} and Theorem~\ref{th2}}
\label{sec:proof-theo-mac}

\subsection{Achievability scheme} \label{secach}
The achievability proof is inspired by the proof for point-to-point channel from \cite{winter2003commitment} and relies on hash challenges, typicality, and privacy amplification \cite{bennett2002generalized},  adapted to a distributed setting.  The achievability proof holds for  Theorem~\ref{th1} and Theorem~\ref{th2}, differing only in the set of allowed input distributions to the channel. To simultaneously capture both theorems, we define
\begin{align*}
\bar{\mathcal{P}}(\mathcal{X}_{\mathcal{L}}) \triangleq 
\begin{cases} \mathcal{P}^{\textup{I}}(\mathcal{X}_{\mathcal{L}}) &  \text{if the bidders are non-colluding}\\
\mathcal{P}(\mathcal{X}_{\mathcal{L}}) & \text{if the bidders are colluding} \end{cases}.
\end{align*}
Fix $p_{X_{\mathcal{L}}}\in \bar{\mathcal{P}}(\mathcal{X}_{\mathcal{L}})$. Define $q_{X_{\mathcal{L}}Y} \triangleq  p_{X_{\mathcal{L}}} p_{Y|X_{\mathcal{L}}}$. Consider $X_{\mathcal{L}}^n$ distributed according to $p^{\otimes n}_{X_{\mathcal{L}}}$.

\noindent{}\textbf{Commit Phase}: Bidder ${\ell} \in\mathcal{L}$ commits to $a_{\ell}$ as follows.
\begin{itemize}
\item Bidder $\ell$ sends  the sequence $X^n_{\ell}$ over the multiple access channel $W$. The verifier observes $Y^n$.
\item The verifier chooses a function $G_{\ell}:\mathcal{X}^{n\mu}_{\ell} \to \{0,1\}^{\eta n} $ at random in a family of two-universal hash functions with $\mu>\eta >0$, and sends $G_{\ell}$ to Bidder $\ell$ over the noiseless channel.
\item Bidder $\ell$ selects a set $\mathcal{S}_{\ell} \subset \llbracket 1 , n \rrbracket$ of size $|\mathcal{S}_{\ell}|=\mu n$ uniformly at random and sends $G_{\ell}(X^n_{\ell}[\mathcal{S}_{\ell}])$ and $\mathcal{S}_{\ell}$  to the verifier over the noiseless channel. Let $T_{\ell}$ be the corresponding sequence observed by the verifier.
\item Bidder $\ell$ sets 
 $\bar{X}_{\ell}^n \triangleq {X}_{\ell}^n[(\bigcup_{{\ell} \in \mathcal{L}} \mathcal{S}_{\ell} )^c]$ and $\bar{n} \triangleq n- |\bigcup_{{\ell} \in \mathcal{L}} \mathcal{S}_{\ell} |$. It then chooses  a function $F_{\ell} : \mathcal{X}^{\bar{n}}_{\ell} \to \{0,1\}^{r_{\ell}}$ at random in a family of two-universal hash functions, and sends $F_{\ell}$ and $E_{l} \triangleq a_{\ell} \oplus F_{\ell}(\bar{X}_{\ell}^n)$ over the noiseless channel.
\end{itemize}
\textbf{Reveal Phase}: Bidder ${\ell} \in\mathcal{L}$ reveals $a_{\ell}$ as follows.
\begin{itemize}
\item Bidder $\ell$ sends $X_{\ell}^n$ and $a_{\ell}$ to the verifier  over the noiseless channel.
\item  The verifier tests that
\begin{enumerate}[(i)]
\item $(X_{\mathcal{L}}^n,Y^n)\in \mathcal{T}_{\epsilon}^n(q_{X_{\mathcal{L}}Y})$;
\item $T_{\ell} = G_{\ell}(X^n_{\ell}[\mathcal{S}_{\ell}]), \forall \ell \in \mathcal{L}$;
\item $a_{\ell} = E_{\ell} \oplus F_{\ell}(\bar{X}^n_{\ell}), \forall \ell \in \mathcal{L}$;
\end{enumerate}
and outputs 1 if all conditions are satisfied, and 0 else.
\end{itemize}

\subsubsection{Definitions}

The following notions of typicality will prove useful. Let $\epsilon>0$.
For $x^n_{\mathcal{L}}\in \mathcal{X}_{\mathcal{L}}^n$, define $\mathcal{T}^n_{W,\epsilon}({x}_{\mathcal{L}}^n)$ and $\mathcal{T}_{\epsilon}^n(q_{X_{\mathcal{L}}Y})$ in \eqref{eqdeftyp1} and \eqref{eqdeftyp2}.

\begin{figure*}[t!]

\begin{align*}
\mathcal{T}^n_{W,\epsilon}({x}_{\mathcal{L}}^n)
& \triangleq \left\{ y^n \in \mathcal{Y}^n : \forall x_{\mathcal{L}}, \forall y,\left|\sum_{i=1}^n \frac{\mathds{1} \{ (x_{\mathcal{L}}, y) \!=\! (x_{\mathcal{L},i}, y_i) \}}{n}  \! - \! W_{x_{\mathcal{L}}}( y) \sum_{i=1}^n \frac{\mathds{1} \{ x_{\mathcal{L}} \!=\! x_{\mathcal{L},i} \}
}{n}  \right|\! \leq\! \epsilon  \right. \\
&\phantom{---}\left. \text{ and } W_{x_{\mathcal{L}}}( y) = 0 \implies \sum_{i=1}^n \frac{\mathds{1} \{ (x_{\mathcal{L}}, y) \!=\! (x_{\mathcal{L},i}, y_i) \}}{n} =0 \right\} . \numberthis \label{eqdeftyp1}\\
\mathcal{T}_{\epsilon}^n(q_{X_{\mathcal{L}}Y}) &\triangleq  \left\{ (x_{\mathcal{L}}^n,y^n ) \in \mathcal{X}_{\mathcal{L}}^n\times \mathcal{Y}^n : \forall x_{\mathcal{L}}, \forall y,  \left| \sum_{i=1}^n \frac{\mathds{1} \{ (x_{\mathcal{L}}, y)\! =\! (x_{\mathcal{L},i}, y_i) \}}{n}   - q_{X_{\mathcal{L}}Y}(x_{\mathcal{L}}, y) \right|\leq \epsilon \right.   \\
&\phantom{---}\left. \text{ and }    q_{X_{\mathcal{L}}Y}(x_{\mathcal{L}}, y) = 0 \implies \sum_{i=1}^n \frac{\mathds{1} \{ (x_{\mathcal{L}}, y)\! = \!(x_{\mathcal{L},i}, y_i) \}}{n} =0 \right\}. \numberthis \label{eqdeftyp2}
\end{align*}
\hrulefill
\end{figure*}

\subsubsection{Correctness}
When the parties are not cheating, standard typicality arguments~\cite{Kramer2008Topics} show that  $\lim_{n\to \infty} \mathbb{P}[(X_{\mathcal{L}}^n,Y^n)\in \mathcal{T}_{\epsilon}^n(q_{X_{\mathcal{L}}Y})] =1$. Consequently, part (i) of the verifier test passes, while part (ii) and (iii) are automatically true, so that the verifier estimates $a_{\mathcal{L}}$ with vanishing probability of error in the reveal phase in the absence of cheating.

\subsubsection{Concealment} \label{secanalaysisconcealment}
Define $V' \triangleq (\mathcal{S}_{\mathcal{L}},G_{\mathcal{L}},F_{\mathcal{L}} ,T_{\mathcal{L}},Y^n)$ and $V\triangleq (V' ,E_{\mathcal{L}} )$, where  $\mathcal{S}_{\mathcal{L}} \triangleq (\mathcal{S}_{\ell})_{\ell \in \mathcal{L}}$, $G_{\mathcal{L}} \triangleq (G_{\ell})_{\ell \in \mathcal{L}}$, $F_{\mathcal{L}} \triangleq (F_{\ell})_{\ell \in \mathcal{L}}$, $T_{\mathcal{L}} \triangleq (T_{\ell})_{\ell \in \mathcal{L}}$, $E_{\mathcal{L}} \triangleq (E_{\ell})_{\ell \in \mathcal{L}}$.  Note that $V$ captures all the information available to the verifier at the end of the commit phase. Define also $K_{\mathcal{L}}  \triangleq (F_{\ell}(\bar{X}^n_{l}))_{\ell \in \mathcal{L}}$, which represents the sequence of hashes used to protect the committed strings by the bidders, and $\bar{Y}^n \triangleq Y^n[ \mathcal{S}^c]$, where $\mathcal{S} \triangleq \bigcup_{{\ell} \in \mathcal{L}} \mathcal{S}_{\ell}$. Then, we have
\begin{align*}
I(A_{\mathcal{L}};V) 
& \stackrel{(a)}= I(A_{\mathcal{L}}; E_{\mathcal{L}}) + I(A_{\mathcal{L}};V'|E_{\mathcal{L}}) \\
& \leq I(A_{\mathcal{L}}; E_{\mathcal{L}}) + I(A_{\mathcal{L}} E_{\mathcal{L}};V') \\
& \stackrel{(b)}= I(A_{\mathcal{L}}; E_{\mathcal{L}}) + I(A_{\mathcal{L}} K_{\mathcal{L}};V') \\
& =  I(A_{\mathcal{L}}; E_{\mathcal{L}}) + I( K_{\mathcal{L}};V') + I(A_{\mathcal{L}};V' |K_{\mathcal{L}}) \\
& \leq  I(A_{\mathcal{L}}; E_{\mathcal{L}}) + I( K_{\mathcal{L}};V') + I(A_{\mathcal{L}};V' K_{\mathcal{L}}) \\
& \stackrel{(c)}=  I(A_{\mathcal{L}}; E_{\mathcal{L}}) + I( K_{\mathcal{L}};V')   \\
& \stackrel{(d)}\leq r_{\mathcal{L}} -H( E_{\mathcal{L}}| A_{\mathcal{L}}) + I( K_{\mathcal{L}};V')   \\
& \stackrel{(e)}= r_{\mathcal{L}} -H( K_{\mathcal{L}}) + I( K_{\mathcal{L}};V') \\
& = r_{\mathcal{L}} -H( K_{\mathcal{L}}) + I( K_{\mathcal{L}};F_{\mathcal{L}} \bar{Y}^n) \\
&\phantom{--}+ I ( K_{\mathcal{L}}; \mathcal{S}_{\mathcal{L}}  G_{\mathcal{L}}T_{\mathcal{L}} Y^n [\mathcal{S}  ]|F_{\mathcal{L}} \bar{Y}^n)\\
& \leq  r_{\mathcal{L}} -H( K_{\mathcal{L}}) + I( K_{\mathcal{L}};F_{\mathcal{L}} \bar{Y}^n) \\
&\phantom{--}+ I ( F_{\mathcal{L}} \bar{Y}^n K_{\mathcal{L}}; \mathcal{S}_{\mathcal{L}}  G_{\mathcal{L}}T_{\mathcal{L}} Y^n [\mathcal{S}])\\
& \stackrel{(f)}=  r_{\mathcal{L}} -H( K_{\mathcal{L}}) + I( K_{\mathcal{L}};F_{\mathcal{L}} \bar{Y}^n) , \numberthis \label{eqconceal}
\end{align*}
where
\begin{enumerate}[(a)]
    \item holds by the chain rule and the definition of $V$;
    \item holds by the definition of $E_{\mathcal{L}}$;
    \item holds by independence between $A_{\mathcal{L}}$ and $(V', K_{\mathcal{L}})$;
    \item holds with $r_{\mathcal{L}} \triangleq \sum_{{\ell} \in\mathcal{L}} r_{\ell}$;
    \item holds by the definition of $E_{\mathcal{L}}$;
    \item holds by independence between $( F_{\mathcal{L}}, \bar{Y}^n, K_{\mathcal{L}})$ and $( \mathcal{S}_{\mathcal{L}}  G_{\mathcal{L}}T_{\mathcal{L}} Y^n [\mathcal{S} ])$.
\end{enumerate}
Then, we upper bound the right-hand side of \eqref{eqconceal} using the version of the distributed leftover hash lemma \cite{wullschleger2007oblivious} in Lemma \ref{lemloh}, and we lower bound the min-entropies \cite{renner2008security,holenstein2011randomness} appearing in  Lemma~\ref{lemloh} using Lemma~\ref{lems1}.
          \begin{lem}[{Distributed leftover hash lemma, e.g., \cite[Lemma~3]{chou2021distributed}}] \label{lemloh}
 Consider a sub-normalized non-negative function  $ p_{ X_{\mathcal L}Z}$ defined over $\bigtimes_{\ell \in \mathcal{L}}\mathcal{X}_{l}\times \mathcal{Z}$, where $X_{\mathcal{L}} \triangleq (X_{\ell})_{\ell \in \mathcal{L}}$ and, $\mathcal{Z}$, $\mathcal{X}_{l}$, $\ell \in\mathcal{L}$, are finite alphabets.    
 For $\ell \in \mathcal{L}$, let  $F_{\ell}:\{0,1\}^{n_{\ell}} \longrightarrow \{0,1\}^{r_{\ell}}$, be uniformly chosen in a family $\mathcal{F}_{\ell}$ of two-universal hash functions. For any $\mathcal{T} \subseteq \mathcal{L}$, define $r_{\mathcal{T}}\triangleq \sum_{i \in \mathcal{T}}r_i$. Define also ${F}_{\mathcal{L}}\triangleq (F_{\ell})_{\ell \in \mathcal{L}}$ and 
   $
        F_{\mathcal{L}}( X_{\mathcal{L}})\triangleq \left( F_{\ell}(X_{\ell})\right)_{{\ell} \in\mathcal{L}}$. Then, for any $q_Z$ defined over $\mathcal{Z}$ such that $\textup{supp}(q_Z) \subseteq \textup{supp}(p_Z)$, we have
   \begin{align}
      \mathbb{V}({{p}_{F_{\mathcal{L}}( X_{\mathcal{L}})  F_{\mathcal{L}}Z}}, p_{U_{\mathcal K}} p_{U_{\mathcal F}} p_Z) \leq   {{{\sqrt{{ \displaystyle\sum_{\substack{{\mathcal T\subseteq\mathcal L}, {\mathcal T \neq \emptyset}}}}\!\!2^{r_{\mathcal T}-H_{\infty}(p_{X_{\mathcal T}Z}|q_Z)}}}}}, \label{eq:lohl}
       \end{align}
      where  $p_{U_{\mathcal K}}$ and  $p_{U_{\mathcal F}}$ are the uniform distributions over $\llbracket 1,2^{r_{{\mathcal{L}}}} \rrbracket$ and $\llbracket 1,\prod_{\ell \in \mathcal{L}} |\mathcal{F}_{\ell}| \rrbracket$, respectively, and  for any $\mathcal T\subseteq\mathcal L, \mathcal T \neq \emptyset$, 
$$H_{\infty}(p_{X_{\mathcal{T}} Z}|q_{Z})\triangleq -\log \displaystyle \max_{\substack{{x_{\mathcal{T}} \in \mathcal{X}_{\mathcal{T}}}\\ z \in \textup{supp}(q_{Z}) }}\frac{p_{X_{\mathcal{T}} Z}(x_{\mathcal{T}}, z)}{q_{Z}(z)}.$$
          \end{lem}

\begin{lem}[{\cite[Lemma 4]{chou2021distributed}}] \label{lems1}
        Let $(\mathcal{X}_{\ell})_{\ell \in \mathcal{L}}$ be $L$ finite alphabets and define for $\mathcal{T}\subseteq \mathcal{L}$, $\mathcal{X}_{\mathcal{T}}\triangleq \bigtimes_{\ell \in \mathcal{T}} \mathcal{X}_{\ell}$. Consider the random variables  $X^{n}_{\mathcal{L}}\triangleq ({X}^{n}_{\ell})_{\ell \in \mathcal{L}}$ and $Z^{n}$ defined over $\mathcal{X}_{\mathcal{L}}^n  \times\mathcal{Z}^n$ with probability distribution $q_{X^{n}_{\mathcal{L}} Z^{n}}\triangleq  q^{\otimes n}_{X_{\mathcal{L}} Z}$. For any $\epsilon>0$, there exists a subnormalized non-negative function $w_{X^{n}_{\mathcal{L}} Z^{n}}$ defined over $\mathcal{X}^n_{\mathcal{L}} \times\mathcal{Z}^n$ such that $\mathbb{V}(q_{X^{n}_{\mathcal{L}} Z^{n}},w_{X^{n}_{\mathcal{L}} Z^{n}})\leq\epsilon$ and
       \begin{align*}
           \forall \mathcal{T}\subseteq \mathcal{L}, H_{\infty}(w_{X^{n}_{\mathcal{T}} Z^{n}}|q_{Z^{n}})\geq n H({X_{\mathcal{T}}}|Z)-n \delta_{\mathcal{T}}(n),
       \end{align*}
       where $\delta_{\mathcal{T}}(n)\triangleq (\log (\lvert\mathcal{X}_{\mathcal{T}}\rvert+3))\sqrt{\frac{2}{n}(L+\log(\frac{1}{\epsilon}))}$. 
\label{lem2}
       \end{lem}

Let $\epsilon>0$. By Lemma \ref{lems1}, there exists a subnormalized non-negative function $w_{\bar{X}^{n}_{\mathcal{L}} \bar{Y}^{n}}$ such that $\mathbb{V}(q_{\bar{X}^{n}_{\mathcal{L}} \bar{Y}^{n}},w_{\bar{X}^{n}_{\mathcal{L}} \bar{Y}^{n}})\leq\epsilon$ and
       \begin{align} \label{eqlmin}
           \forall \mathcal{T}\subseteq \mathcal{L}, H_{\infty}(w_{\bar{X}^{n}_{\mathcal{T}} \bar{Y}^{n}}|q_{\bar{Y}^{n}})\geq \bar{n} H({X_{\mathcal{T}}}|Y)-\bar{n} \delta_{\mathcal{T}}(\bar{n}),
       \end{align}
       where $\delta_{\mathcal{T}}(\bar{n})\triangleq (\log (\lvert\mathcal{X}_{\mathcal{T}}\rvert+3))\sqrt{\frac{2}{\bar{n}}(L+\log(\frac{1}{\epsilon}))}$. Then, we~have
\begin{align*}
&\mathbb{V} ( q_{K_{\mathcal{L}}F_{\mathcal{L}} \bar{Y}^n}, p_{U_{\mathcal K}} p_{U_{\mathcal F}}  q_{\bar{Y}^n})\\
& \stackrel{(a)}\leq \mathbb{V} ( q_{K_{\mathcal{L}}F_{\mathcal{L}} \bar{Y}^n}, w_{K_{\mathcal{L}}F_{\mathcal{L}} \bar{Y}^n}) + \mathbb{V} ( w_{K_{\mathcal{L}}F_{\mathcal{L}} \bar{Y}^n}, p_{U_{\mathcal K}} p_{U_{\mathcal F}}  q_{\bar{Y}^n}) \\
& \stackrel{(b)}\leq \epsilon + \mathbb{V} ( w_{K_{\mathcal{L}}F_{\mathcal{L}} \bar{Y}^n}, p_{U_{\mathcal K}} p_{U_{\mathcal F}}  q_{\bar{Y}^n}) \\
& \stackrel{(c)}\leq \epsilon + \mathbb{V} ( w_{K_{\mathcal{L}}F_{\mathcal{L}} \bar{Y}^n}, p_{U_{\mathcal K}} p_{U_{\mathcal F}}  w_{\bar{Y}^n}) \\
& \phantom{---}+ \mathbb{V} ( p_{U_{\mathcal K}} p_{U_{\mathcal F}}  w_{\bar{Y}^n}, p_{U_{\mathcal K}} p_{U_{\mathcal F}}  q_{\bar{Y}^n}) \\
& \stackrel{(d)}\leq 2 \epsilon + \mathbb{V} ( w_{K_{\mathcal{L}}F_{\mathcal{L}} \bar{Y}^n}, p_{U_{\mathcal K}} p_{U_{\mathcal F}}  w_{\bar{Y}^n}) \\
& \stackrel{(e)}\leq 2 \epsilon +    {{{\sqrt{{ \displaystyle\sum_{\substack{{\mathcal T \subseteq\mathcal L}, {\mathcal T \neq \emptyset}}}}2^{r_{\mathcal T}-H_{\infty}(w_{\bar{X}^{n}_{\mathcal{T}} \bar{Y}^{n}}|q_{\bar{Y}^{n}})}}}}}\\
& \stackrel{(f)}\leq 2 \epsilon +    {{{\sqrt{{ \displaystyle\sum_{\substack{{\mathcal T\subseteq\mathcal L}, {\mathcal T \neq \emptyset}}}}2^{r_{\mathcal T}-\bar{n} H({X_{\mathcal{T}}}|Y)+\bar{n} \delta_{\mathcal{T}}(\bar{n})}}}}}, \numberthis \label{eqVar}
\end{align*}
where
\begin{enumerate}
    \item[(a)] and (c) hold by the triangle inequality;
    \item[(b)] and (d) hold by the data processing inequality and because $\mathbb{V}(q_{\bar{X}^{n}_{\mathcal{L}} \bar{Y}^{n}},w_{\bar{X}^{n}_{\mathcal{L}} \bar{Y}^{n}})\leq\epsilon$;
    \item[(e)] holds by Lemma \ref{lemloh};
    \item[(f)] holds by~\eqref{eqlmin}.
\end{enumerate} 
Finally, we conclude by invoking the following lemma.
\begin{lem}[\cite{csiszar1996almost}] \label{lemvarmi}
Let $X$ and $Y$ be discrete random variables defined over $\mathcal{X}$ and $\mathcal{Y}$, respectively, with joint distribution $p_{XY}$. If $|\mathcal{X}|\geq 4$, then 
$$
 I(X;Y) \leq \mathbb{V}(p_{XY},p_Xp_Y) \log \frac{|\mathcal{X}|}{\mathbb{V}(p_{XY},p_Xp_Y)}.
$$
\end{lem}
By \eqref{eqVar} and Lemma \ref{lemvarmi}, concealment holds with $(r_{l})_{{\ell} \in\mathcal{L}}$ such that for any $\mathcal{T} \subseteq \mathcal{L}$, $ \lim_{n\to\infty} \frac{r_{\mathcal{T}}}{n} \leq  H(X_{\mathcal{T}}|Y) -\delta$, $\delta>0$.

\subsubsection{Bindingness}

We will use the following lemma from \cite{winter2003commitment}.

\begin{lem}[Adapted from {\cite{winter2003commitment}}] \label{lem1}
Let $\delta, \sigma>0$. Consider $x^n_{\mathcal{L}},\tilde{x}^n_{\mathcal{L}} \in \mathcal{X}_{\mathcal{L}}^n$ such that $d_H(x^n_{\mathcal{L}},\tilde{x}^n_{\mathcal{L}}) \geq \sigma n $ and  a non-redundant multiple access channel $W$. 
Then, 
$$
 \lim_{n\to\infty}  W^{\otimes n}_{x^n_{\mathcal{L}}}( \mathcal{T}^n_{W,\epsilon}(\tilde{x}_{\mathcal{L}}^n)) =0.
$$
\end{lem}
A close inspection of the proof reveals that the lemma holds regardless of whether the joint distribution of the inputs is a product of the marginals or not.

In the reveal phase, if the verifier observes $\tilde{x}^n_{\mathcal{L}}$, then, by Lemma \ref{lem1}, a successful joint typicality test at the verifier requires $d_H(\tilde{x}^n_{\mathcal{L}}, {x}^n_{\mathcal{L}}) \leq {O(n^\alpha)}$, for some $\alpha<1$. This implies that Test~(ii) at the verifier in the reveal phase  can only succeed with a probability at most 
\begin{align*}
    2^{-Ln\eta}\sum_{i=0}^{d_H(\tilde{x}^n_{\mathcal{L}}, {x}^n_{\mathcal{L}})}{n \choose i}
    &\leq 2^{-Ln\eta} \cdot 2^{H_b(O(n^{\alpha-1}))}\\
    &\leq  2^{-Ln\eta} \cdot 2^{O(n^\alpha \log n)}\\
    & \xrightarrow{n\to \infty } 0,
\end{align*} 
where $H_b$ denotes the binary entropy.

\subsubsection{Achievable region and sum-rate}
For any   $p_{X_{\mathcal{L}}}\in \bar{\mathcal{P}}(\mathcal{X}_{\mathcal{L}})$, we have shown the achievability~of 
\begin{align*}
\mathcal{R} (p_{X_{\mathcal{L}}}) \triangleq  
\{ (R_{\ell})_{{\ell} \in\mathcal{L}} : R_{\mathcal{T}}  \leq  H(X_{\mathcal{T}}|Y) , \forall\mathcal{T}\subseteq\mathcal{L}\}.
\end{align*}
Next, define the set function
\begin{align*}
f_{p_{X_{\mathcal{L}}}} : 2^{\mathcal{L}} & \to \mathbb{R} \\
\mathcal{T} & \mapsto H(X_{\mathcal{T}}|Y),
\end{align*}
where $2^{\mathcal{L}}$ denotes the power set of $\mathcal{L}$. $f_{p_{X_{\mathcal{L}}}}$  is normalized, i.e., $f_{p_{X_{\mathcal{L}}}}(\emptyset)=0$, non-decreasing, i.e., $\forall \mathcal{S}, \mathcal{T} \subseteq \mathcal{L}$, $\mathcal{S}\subseteq \mathcal{T} \implies f_{p_{X_{\mathcal{L}}}}(\mathcal{S})\leq f_{p_{X_{\mathcal{L}}}}(\mathcal{T})$, and submodular because for $\mathcal{U},\mathcal{V} \subseteq \mathcal{L}$, we have
\begin{align*}
&f_{p_{X_{\mathcal{L}}}} (\mathcal{U} \cup \mathcal{V}) + f_{p_{X_{\mathcal{L}}}} (\mathcal{U} \cap \mathcal{V})\\
& =  H(X_{\mathcal{U}}|Y) +   H(X_{\mathcal{V} \backslash \mathcal{U}}|Y X_{\mathcal{U}}) + H(X_{\mathcal{U} \cap \mathcal{V}} | Y)  \\
& =  H(X_{\mathcal{U}}|Y) +   H(X_{\mathcal{V} \backslash \mathcal{U}}|Y X_{\mathcal{U}}) \\
&\phantom{--}+ H(X_{\mathcal{V}}|Y)-   H(X_{\mathcal{V} \backslash \mathcal{U}}|YX_{\mathcal{U} \cap \mathcal{V}} )\\
& \leq H(X_{\mathcal{U}}|Y) + H(X_{\mathcal{V}}|Y) \\
& = f_{p_{X_{\mathcal{L}}}} (\mathcal{U} ) + f_{p_{X_{\mathcal{L}}}} (  \mathcal{V}),
\end{align*}
where the inequality holds because $ H(X_{\mathcal{V} \backslash \mathcal{U}}|Y X_{\mathcal{U}}) \leq   H(X_{\mathcal{V} \backslash \mathcal{U}}|YX_{\mathcal{U} \cap \mathcal{V}} )$ since conditioning reduces entropy. 

Hence, by \cite{edmonds1970submodular}, the rate-tuple $(f_{p_{X_{\mathcal{L}}}} (\llbracket l,L \rrbracket) - f_{p_{X_{\mathcal{L}}}} (\llbracket l+1,L \rrbracket))_{{\ell} \in\mathcal{L}}$ is achievable and so is the sum-rate $f_{p_{X_{\mathcal{L}}}} (\llbracket 1,L \rrbracket)=H(X_{\mathcal{L}}|Y)$. Hence, the following sum-rate is achievable $$R_{\mathcal{L}} = \max_{p_{X_{\mathcal{L}} \in \bar{\mathcal{P}}(\mathcal{X}_{\mathcal{L}})}  } H(X_{\mathcal{L}}|Y).$$

\subsection{Converse} \label{secconv}
We provide here the converse proof for Theorem \ref{th2}. This proof, particularly Lemma \ref{lem}, draws on techniques from \cite{winter2003commitment,hayashi2022commitment}. The converse proof for the sum-rate in Theorem \ref{th1} follows in a similar manner. We note that we do not obtain a complete characterization of the capacity region when bidders collude, as Lemma \ref{lemmarkov} below is only valid for the non-colluding bidders case.

\begin{lem} \label{lemmarkov}
Consider the non-colluding bidders case. For $\ell \in \mathcal{L}$, $(A_{\ell},S_{\ell})-(M_{\ell},X_{\ell}^n)-(Y^n,S'_{\ell})$ forms a Markov chain.
\end{lem}
\begin{proof}
For ${\ell} \in\mathcal{L}$, $i \in \llbracket 1 , n \rrbracket$, $j \in \llbracket 1 , r_i \rrbracket$, define $\bar{M}_{l,1:i,1:j} \triangleq (M_{l,1:i,1:j},M'_{l,1:i,1:j})$ and $\bar{M}_{l,1:i} \triangleq \bar{M}_{l,1:i,1:r_i}$.
We have
\begin{align*}
    &I(A_{\ell}S_{\ell};Y^n S'_{\ell} | \bar{M}_{l,1:n}X^n_{\ell}) \numberthis \label{eqr1}\\
    & = I(A_{\ell}S_{\ell};Y^n S'_{\ell} | \bar{M}_{l,1:n-1} \bar{M}_{l,n,1:r_n}X^n_{\ell}) \\
    & = I(A_{\ell}S_{\ell};Y^n S'_{\ell} | \bar{M}_{l,1:n-1} \bar{M}_{l,n,1:r_n-1} M_{l,n,r_n}M'_{l,n,r_n} X^n_{\ell}) \\
    & \leq I(A_{\ell}S_{\ell};Y^n S'_{\ell} M'_{l,n,r_n}| \bar{M}_{l,1:n-1} \bar{M}_{l,n,1:r_n-1} M_{l,n,r_n}  X^n_{\ell}) \\
    & \stackrel{(a)}= I(A_{\ell}S_{\ell};Y^n S'_{\ell} | \bar{M}_{l,1:n-1} \bar{M}_{l,n,1:r_n-1} M_{l,n,r_n}  X^n_{\ell}) \\
    & \leq I(A_{\ell}S_{\ell}M_{l,n,r_n};Y^n S'_{\ell} | \bar{M}_{l,1:n-1} \bar{M}_{l,n,1:r_n-1} X^n_{\ell}) \\
     & \stackrel{(b)} = I(A_{\ell}S_{\ell};Y^n S'_{\ell} | \bar{M}_{l,1:n-1} \bar{M}_{l,n,1:r_n-1} X^n_{\ell}) \numberthis \label{eqr2}\\
        & \stackrel{(c)}\leq I(A_{\ell}S_{\ell};Y^n S'_{\ell} | \bar{M}_{l,1:n-1}  X^n_{\ell}) \\
    & \stackrel{(d)}= I(A_{\ell}S_{\ell};Y^{n-1} S'_{\ell} | \bar{M}_{l,1:n-1}   X^n_{\ell}) \\
    & \leq I(A_{\ell}S_{\ell} (X_{l})_n ;Y^{n-1} S'_{\ell} | \bar{M}_{l,1:n-1}   X^{n-1}_{\ell})\displaybreak[0] \\
    & \stackrel{(e)} = I(A_{\ell}S_{\ell} ;Y^{n-1} S'_{\ell} | \bar{M}_{l,1:n-1}   X^{n-1}_{\ell}) \numberthis \label{eqr3}\displaybreak[0]\\
    & \stackrel{(f)}\leq I(A_{\ell}S_{\ell} ; S'_{\ell}  )\displaybreak[0]\\
    & = 0,
\end{align*}
where 
\begin{enumerate}[(a)]
    \item holds because $M'_{l,n,r_n}$ is a function of $(S'_{\ell}, M_{l,1:n,1:r_n},Y^n)$;
    \item holds because $M_{l,n,r_n}$ is a function of $(A_{\ell},S_{\ell}, M'_{l,1:n,1:r_n-1})$;
    \item holds by repeating $r_n-1$ times the steps between \eqref{eqr1} and \eqref{eqr2};
    \item holds because $Y_n - ( \bar{M}_{l,1:n-1},X^n, Y^{n-1} ,S'_{\ell})-(A_{\ell},S_{\ell})$ forms a Markov chain;
    \item holds because $(X_{l})_n$ is a function of $(A_{\ell},S_{\ell}, M'_{l,1:n-1,1:r_{n-1}})$;
    \item holds by repeating $n-1$ times the steps between \eqref{eqr1} and~\eqref{eqr3}.
\end{enumerate}
\end{proof}

\begin{lem} \label{lem}
For the non-colluding case,  there exist $\hat{A}_{l}(V_{\mathcal{L}},  X_{l}^n)$, $\ell \in \mathcal{L}$, such that $$ \lim_{n\to\infty}    \mathbb{P}[ \hat{A}_{l}(V_{\mathcal{L}},  X_{l}^n) \neq {A}_{l}]  = 0, \forall \ell \in \mathcal{L}.$$
\end{lem}
\begin{proof}
Let $\ell \in \mathcal{L}$ and $\delta,\gamma>0$.  We suppose that Bidder $\ell$ behaves honestly during the commit phase. Define for any ($a_{\ell}$, $s_{\ell}$, $x_{\ell}^n$, $m_{\ell}$)
\begin{align}
f(a_{\ell},s_{\ell}) &\triangleq \mathbb{E}_{Y^nM_{\ell}S'_{\ell}|A_{\ell} = a_{\ell}, S_{\ell}=s_{\ell}} [\mathds{1}\{ \beta_{\ell} ( Y^n\!,M_{\ell},S'_{\ell},a_{\ell},s_{\ell})\}],   \label{eq0}\\
\mathcal{G}(a_{\ell}) &\triangleq \{ s_{\ell} : f(a_{\ell},s_{\ell}) > 1- \gamma\}, \label{eq00}
\end{align}
\begin{align}
&F(x_{\ell}^n, m_{\ell} |a_{\ell}, s_{\ell}) \nonumber \\ &\triangleq \mathbb{E}_{Y^n S'_{\ell}| M_{\ell} = m_{\ell}, X^n_{\ell} = x^n_{\ell}, S_{\ell}=s_{\ell}} [\mathds{1}\{ \beta_{\ell} ( Y^n,m_{\ell},S'_{\ell},a_{\ell},s_{\ell})\}] , \label{eq1}
\end{align}
\begin{align}
F(x_{\ell}^n, m_{\ell} |a_{\ell})  &\triangleq  \max_{s_{\ell} \in \mathcal{G}(a_{\ell})} F(x_{\ell}^n, m_{\ell}|a_{\ell}, s_{\ell}),\label{eq2}\\
\bar{F}(x_{\ell}^n, m_{\ell} |a_{\ell})  &\triangleq  \max_{a_{\ell}^* \neq a_{\ell}} F(x_{\ell}^n, m_{\ell}|a_{\ell}^*),
\end{align}
and consider $$\hat{a}_{\ell}(x^n_{\ell},m_{\ell}) \in \argmax_{a_{\ell}} F(x^n_{\ell}, m_{\ell} | a_{\ell}).$$ 
Then,    we have
\begin{align*}
&\mathbb{P}[\hat{a}_{\ell}(X^n_{\ell},M_{\ell}) \neq a_{\ell}]\\
& = \mathbb{E}_{X_{\ell}^n M_{\ell}|A_{\ell} = a_{\ell}} [\mathds{1} \{ \hat{a}_{\ell}(X^n_{\ell},M_{\ell}) \neq a_{\ell} \} ] \\
& \leq \mathbb{E}_{X_{\ell}^n M_{\ell}|A_{\ell} = a_{\ell}} [  \bar{F}(X_{\ell}^n, M_{\ell} |a_{\ell})  + 1- {F}(X_{\ell}^n, M_{\ell} |a_{\ell}) ]\\
& = \mathbb{E}_{X_{\ell}^n M_{\ell}|A_{\ell} = a_{\ell}}   \bar{F}(X_{\ell}^n, M_{\ell} |a_{\ell}) \\
& \phantom{--}+ \mathbb{E}_{X_{\ell}^n M_{\ell}|A_{\ell} = a_{\ell}}[1- {F}(X_{\ell}^n, M_{\ell} |a_{\ell}) ], \numberthis \label{eqF}
\end{align*}
where the inequality holds because if $\hat{a}_{\ell}(x^n_{\ell},m_{\ell}) \neq a_{\ell}$, then ${F}(x_{\ell}^n, m_{\ell} |a_{\ell}) \leq \max_{a_{\ell}^* \neq a_{\ell}} F(x^n_{\ell}, m_{\ell} | a_{\ell}^*) = \bar{F}(x_{\ell}^n, m_{\ell} |a_{\ell})$, so~that
$
\mathds{1} \{ \hat{a}_{\ell}(x^n_{\ell},m_{\ell}) \neq a_{\ell} \} 
 \leq \bar{F}(x_{\ell}^n, m_{\ell} |a_{\ell})  + 1- {F}(x_{\ell}^n, m_{\ell} |a_{\ell})$. 
 
 We next upper bound the second term in the right-hand side of \eqref{eqF}. We have
\begin{align*}
& \mathbb{E}_{X_{\ell}^n M_{\ell} | A_{\ell} = a_{\ell} }  F(X_{\ell}^n, M_{\ell}|a_{\ell}) \\
& \stackrel{(a)} = \mathbb{E}_{X_{\ell}^n M_{\ell} S_{\ell} | A_{\ell} = a_{\ell} }  F(X_{\ell}^n, M_{\ell}|a_{\ell}) \\
&\geq \textstyle\sum_{s_{\ell} \in \mathcal{G}(a_{\ell})} p_{S_{\ell} | A_{\ell} = a_{\ell}} (s_{\ell}) \\
& \phantom{----}\times\mathbb{E}_{X_{\ell}^n M_{\ell} | A_{\ell} = a_{\ell} , S_{\ell} = s_{\ell}}  F(X_{\ell}^n, M_{\ell}|a_{\ell}) \\
&\stackrel{(b)} \geq \textstyle\sum_{s_{\ell} \in \mathcal{G}(a_{\ell})} p_{S_{\ell} | A_{\ell} = a_{\ell}} (s_{\ell}) \\
& \phantom{----}\times\mathbb{E}_{X_{\ell}^n M_{\ell} | A_{\ell} = a_{\ell} , S_{\ell} = s_{\ell}}  F(X_{\ell}^n, M_{\ell}|a_{\ell},s_{\ell}) \\
   & \stackrel{(c)}= \textstyle\sum_{s_{\ell} \in \mathcal{G}(a_{\ell})} p_{S_{\ell} | A_{\ell} = a_{\ell}} (s_{\ell})  \\
& \phantom{----}\times\mathbb{E}_{  Y^n S'_{\ell} M_{\ell} |  S_{\ell}=s_{\ell}, A_{\ell} =a_{\ell}} [\mathds{1}\{ \beta ( Y^n,M_{\ell}, S'_{\ell},a_{\ell},s_{\ell})\}] \\
&\stackrel{(d)}= \textstyle\sum_{s_{\ell} \in \mathcal{G}(a_{\ell})} p_{S_{\ell} | A_{\ell} = a_{\ell}} (s_{\ell})f(a_{\ell},s_{\ell}) \\
&\stackrel{(e)}> \textstyle\sum_{s_{\ell} \in \mathcal{G}(a_{\ell})} p_{S_{\ell} | A_{\ell} = a_{\ell}} (s_{\ell}) (1-\gamma) \\
&\stackrel{(f)}\geq (1 - \delta \gamma^{-1})
 (1-\gamma), \numberthis \label{eqf1}
\end{align*}
where 
\begin{enumerate}[(a)]
    \item holds by marginalization;
    \item  holds by \eqref{eq2};
    \item holds because, for any ($a_{\ell}^{\dagger}$,  $s_{\ell}^{\dagger}$), we have
\begin{align*}
&  \!\!\!\!\! \mathbb{E}_{X_{\ell}^n M_{\ell} | A_{\ell} = a_{\ell}, S_{\ell}=s_{\ell}} F(X_{\ell}^n, M_{\ell}|a_{\ell}^{\dagger}, s_{\ell}^{\dagger})\\
 & \!\!\!\!\! \stackrel{(i)}= \mathbb{E}_{X_{\ell}^n M_{\ell} | A_{\ell} = a_{\ell}, S_{\ell}=s_{\ell}}\mathbb{E}_{Y^n S'_{\ell}|M_{\ell} = M_{\ell}, X^n_{\ell} = X^n_{\ell}, S_{\ell}=s_{\ell}}\\
 &\phantom{-------l-----} [\mathds{1}\{ \beta_{\ell} ( Y^n,M_{\ell},S'_{\ell},a_{\ell}^{\dagger},s_{\ell}^{\dagger})\}] \\
  & \!\!\!\!\!\!\stackrel{(ii)}= \mathbb{E}_{X_{\ell}^n M_{\ell} | A_{\ell} = a_{\ell}, S_{\ell}=s_{\ell}}\mathbb{E}_{Y^n S'_{\ell}|M_{\ell} = M_{\ell}, X^n_{\ell} = X^n_{\ell}, S_{\ell}=s_{\ell}, A_{\ell} =a_{\ell}} \\
 &\phantom{-------l-----}[\mathds{1}\{ \beta_l ( Y^n,M_{\ell},S'_{\ell},a_{\ell}^{\dagger},s_{\ell}^{\dagger})\}] \\
    & \!\!\!\!\!=   \mathbb{E}_{X_{\ell}^n  Y^n S'_{\ell} M_{\ell} |  S_{\ell}=s_{\ell}, A_{\ell} =a_{\ell}} [\mathds{1}\{ \beta_l ( Y^n,M_{\ell},S'_{\ell},a_{\ell}^{\dagger},s_{\ell}^{\dagger})\}]\\
        & \!\!\!\!\!\!\!\stackrel{(iii)}=   \mathbb{E}_{  Y^n S'_{\ell} M_{\ell} |  S_{\ell}=s_{\ell}, A_{\ell} =a_{\ell}} [\mathds{1}\{ \beta_l ( Y^n,M_{\ell}, S'_{\ell},a_{\ell}^{\dagger},s_{\ell}^{\dagger})\}], \numberthis \label{eq3}
\end{align*}
where 
\begin{enumerate}[(i)]
    \item holds by \eqref{eq1};
    \item holds because  $A_{\ell}- (M_{\ell} , X^n_{\ell} , S_{\ell})-( Y^n ,S'_{\ell}) $ forms a Markov chain since $I(A_{\ell};  Y^n S'_{\ell}|M_{\ell}  X^n_{\ell}  S_{\ell}  ) \leq I(A_{\ell} S_{\ell};  Y^n S'_{\ell}|M_{\ell}  X^n_{\ell}    ) =0$ by Lemma~ \ref{lemmarkov};
    \item  holds by marginalization over $X^n_{\ell}$.
\end{enumerate}
    \item holds by \eqref{eq0};
    \item holds by \eqref{eq00};
    \item holds because \begin{align*}
\mathbb{P}[\mathcal{G}(a_{\ell})]
& = \mathbb{P}[f(a_{\ell},S_{\ell})> 1- \gamma ]\\
& = 1 - \mathbb{P}[1- f(a_{\ell},S_{\ell}) \geq  \gamma ]\\
& \stackrel{(i)} \geq 1 - \frac{\mathbb{E}_{S_{\ell}|A_{\ell} = a_{\ell}}[1- f(a_{\ell},S_{\ell}) ]}{\gamma} \\
& \stackrel{(ii)} \geq 1 - \delta \gamma^{-1},  
\end{align*}
where 
\begin{enumerate}[(i)]
    \item holds by Markov's inequality;
    \item  holds because by the correctness condition, for any $a_{\ell}$ and for $n$ large enough, we have $
\mathbb{E}_{S_{\ell}|A_{\ell} = a_{\ell}} f(a_{\ell},S_{\ell})  \geq 1 -\delta.$
\end{enumerate}
\end{enumerate}

 We now upper bound the first term in the right-hand side of \eqref{eqF}.  We have
\begin{align*}
& \mathbb{E}_{X_{\ell}^n  M_{\ell} |   A_{\ell} =a_{\ell}} [\bar{F}(X_{\ell}^n, M_{\ell} |a_{\ell})]\\
& \stackrel{(a)} =\mathbb{E}_{X_{\ell}^n  M_{\ell} |   A_{\ell} =a_{\ell}} [{F}(X_{\ell}^n, M_{\ell} |a_{\ell}^*,s_{\ell}^*)] \\
&  \stackrel{(b)}= \mathbb{E}_{S_{\ell} |   A_{\ell} =a_{\ell}} \mathbb{E}_{X_{\ell}^n  M_{\ell} |   A_{\ell} =a_{\ell}, S_{\ell}=S_{\ell}} [{F}(X_{\ell}^n, M_{\ell} |a_{\ell}^*,s_{\ell}^*)] \\
&  \stackrel{(c)}= \mathbb{E}_{S_{\ell} |   A_{\ell} =a_{\ell}} \mathbb{E}_{  Y^n S'_{\ell} M_{\ell} |  S_{\ell}=S_{\ell}, A_{\ell} =a_{\ell}} [\mathds{1}\{ \beta_l ( Y^n\!\!,M_{\ell}, S'_{\ell},a_{\ell}^*,s_{\ell}^*)\!\}] \\
&  \stackrel{(d)}\leq \delta,  \numberthis \label{eqf2}
\end{align*}
where 
\begin{enumerate}[(a)]
    \item holds with $(a_{\ell}^*, s_{\ell}^*)$ such that $\bar{F}(X_{\ell}^n, M_{\ell} |a_{\ell}) = {F}(X_{\ell}^n, M_{\ell} |a_{\ell}^*) = {F}(X_{\ell}^n, M_{\ell} |a_{\ell}^*,s_{\ell}^*)$; 
    \item holds by marginalization over $S_{\ell}$; 
    \item holds by \eqref{eq3};
    \item holds by the bindingness condition. 
\end{enumerate}
Finally, by \eqref{eqF}, \eqref{eqf1}, and \eqref{eqf2}, we have
$$ \mathbb{P}[\hat{a}_{\ell}(X^n_{\ell},M_{\ell}) \neq a_{\ell}]\xrightarrow{n\to\infty} 0.$$
Since $M_{\ell}$ is part of $V_{\mathcal{L}}$,  the lemma follows.
\end{proof}

Finally, for $U$ uniformly distributed over $\llbracket 1 ,n \rrbracket$ and independent of all other random variables, for any $\mathcal{T} \subseteq \mathcal{L}$, we~have
\begin{align*}
 n H(X_{\mathcal{T}} |Y ) 
& = n H(X_{\mathcal{T},U} |Y_U ) \\
& \stackrel{(a)}\geq n H(X_{\mathcal{T},U} |Y_U U) \\
& = \sum_{i=1}^n H(X_{\mathcal{T},i} |Y_i) \\
& \stackrel{(b)}\geq \sum_{i=1}^n H(X_{\mathcal{T},i} |Y^n  X^{i-1}_{\mathcal{T}}) \\
&  \stackrel{(c)}=   H(X^n_{\mathcal{T}} |Y^n ) \\
&  \stackrel{(d)}\geq   H(X^n_{\mathcal{T}} |V_{\mathcal{L}} ) \\
& =    H(X^n_{\mathcal{T}} {A}_{\mathcal{T}} |V_{\mathcal{L}}) - H({A}_{\mathcal{T}} | X^n_{\mathcal{T}} V_{\mathcal{L}} )\\
&  \stackrel{(e)}\geq     H(X^n_{\mathcal{T}} {A}_{\mathcal{T}} |V_{\mathcal{L}}) - H({A}_{\mathcal{T}} | \hat{A}_{\mathcal{T}} )\\
& \stackrel{(f)}\geq   H( {A}_{\mathcal{T}} |V_{\mathcal{L}} ) - o(n)\\
& =  H( {A}_{\mathcal{T}}) - I( {A}_{\mathcal{T}};V_{\mathcal{L}} ) - o(n)\\
& \stackrel{(g)}\geq   H( {A}_{\mathcal{T}})  - o(n)\\
& =   nR_{\mathcal{T}}  - o(n), \numberthis \label{eqconvrate}
\end{align*}
where 
\begin{enumerate}
    \item[(a)]\!\!, (b), and (d) hold because conditioning reduces entropy;
    \item[(c)] holds by the chain rule;
    \item[(e)] holds with $\hat{A}_{\mathcal{T}} \triangleq (\hat{A}_{l})_{{\ell} \in\mathcal{T}}$ from Lemma \ref{lem} and the data processing inequality; \item[(f)]~holds by Lemma \ref{lem} and Fano's inequality; \item[(g)] holds by the concealment requirement. 
    \end{enumerate}

\section{Proof of Theorem~\ref{thm3}}
\label{sec:proof-theorem-single-bid-multi-ver}

\subsection{Achievability} \label{secach2}

Fix $p_{X}\in {\mathcal{P}}(\mathcal{X})$. Define $q_{XY_{\mathcal{B}}} \triangleq  p_{X } p_{Y_{\mathcal{B}}|X}$. Consider $X^n$ distributed according to $p^{\otimes n}_{X}$.

\noindent{}\textbf{Commit Phase}: The bidder commits to $a$ as follows.
\begin{itemize}
\item The bidder sends  the sequence $X^n$ over the channel $W$. Verifier $b\in\mathcal{B}$ observes $Y_b^n$.
\item Verifier $b\in\mathcal{B}$ chooses a function $G_{b}:\mathcal{X}^{n\mu} \to \{0,1\}^{\eta n} $ at random in a family of two-universal hash functions with $\mu>\eta >0$, and sends $G_{b}$ to the bidder  over the noiseless channel.
\item The bidder selects a set $\mathcal{S} \subset \llbracket 1 , n \rrbracket$ with size $|\mathcal{S}| = \mu n$ uniformly at random and then sends $G_{b}(X^n[\mathcal{S}])$ and $\mathcal{S}$   to Verifier $b \in \mathcal{B}$ over the noiseless channel. Let $T_{b}$ be the corresponding sequence observed by Verifier~$b$.
\item The bidder chooses  a function $F : \mathcal{X}^{\bar{n}} \to \{0,1\}^{r}$ at random in a family of two-universal hash functions, and sends $F$ and $E \triangleq a \oplus F(\bar{X}^n)$ over the noiseless channel, where $\bar{X}^n \triangleq {X}^n[\mathcal{S}^c]$ and $\bar{n} \triangleq n- |  \mathcal{S}  |$.
\end{itemize}
\textbf{Reveal Phase}: Suppose that $\mathcal{A} \subset \mathcal{B}$ is the set of available verifiers with $|\mathcal{A}|\geq 1$. The bidder chooses an arbitrary index $b^{\star} \in \mathcal{B}$ and the bidder  reveals $a$ as follows.
\begin{itemize}
\item The bidder sends $X^n$ and $a$ to the verifier over the noiseless channel.
\item  The verifier tests that
\begin{enumerate}[(i)]
\item $(X^n,Y_{b^{\star}}^n)\in \mathcal{T}_{\epsilon}^n(q_{XY_{b^{\star}}})$;
\item $T_{b^{\star}} = G_{b^{\star}}(X^n[\mathcal{S}])$;
\item $a = E \oplus F(\bar{X}^n)$;
\end{enumerate}
and outputs 1 if all conditions are satisfied, and 0 else.
\end{itemize}

Because of the similarity with the achievability proof of Theorem \ref{th1}, we only highlight the main steps of the proof.

\subsubsection{Correctness}
When the parties are not cheating, standard typicality arguments~\cite{Kramer2008Topics} show that, for any $b\in  \mathcal{B}$,  $ \lim_{n\to \infty} \mathbb{P}[(X^n,Y_{b})\in \mathcal{T}_{\epsilon}^n(q_{X Y_{b}})]=1$. Consequently, part (i) of the reveal phase test passes, while part (ii) and (iii) are automatically true, so that the verifier estimates $a$ with vanishing probability of error in the reveal phase in the absence of cheating.

\subsubsection{Concealment}
Let $b \in \mathcal{B}$. Define $V_{b}' \triangleq (\mathcal{S},G_{\mathcal{B}} , F , T_{\mathcal{B}},Y_{b}^n)$, where $G_{\mathcal{B}} \triangleq (G_{b})_{b\in\mathcal{B}}$ and $T_{\mathcal{B}} \triangleq (T_{b})_{b\in\mathcal{B}}$, and $V_{b} \triangleq (V_{b}' ,E )$.
Also define $K   \triangleq F (\bar{X}^n ) $ and $\bar{Y}_{b}^n \triangleq Y_{b}^n[ \mathcal{S}^c]$. Then, we have
\begin{align}
I(A;V_{b}) \leq  r  -H( K ) + I( K ;F  \bar{Y}_{b}^n) ,  \label{eqconceal2}
\end{align}
where \eqref{eqconceal2} can be shown similar to  \eqref{eqconceal}.
Next, we upper bound the right-hand side of \eqref{eqconceal2} using  Lemmas~\ref{lemloh} and \ref{lems1}.

Let $\epsilon>0$. By Lemma \ref{lems1}, there exists a subnormalized non-negative function $w_{\bar{X}^{n} \bar{Y}_{b}^{n}}$ such that $\mathbb{V}(q_{\bar{X}^{n} \bar{Y}_{b}^{n}},w_{\bar{X}^{n} \bar{Y}_{b}^{n}})\leq\epsilon$ and
       \begin{align} \label{eqlmin2}
           H_{\infty}(w_{\bar{X}^{n} \bar{Y}_{b}^{n}}|q_{\bar{Y}_{b}^{n}})\geq \bar{n} H(X|Y_{b})-\bar{n} \delta(\bar{n}),
       \end{align}
       where $\delta(\bar{n})\triangleq (\log (\lvert\mathcal{X}\rvert+3))\sqrt{\frac{2}{\bar{n}}(1+\log(\frac{1}{\epsilon}))}$. Then, similar to \eqref{eqVar},  using Lemma \ref{lemloh}, one can show that, for any $b\in\mathcal{B}$,
\begin{align}
&\mathbb{V} ( q_{KF \bar{Y}_{b}^n}, p_{U_{\mathcal K}} p_{U_{\mathcal F}}  q_{\bar{Y}_{b}^n}) \nonumber\\
&\leq 2 \epsilon +    \sqrt{2^{r-\bar{n} H(X|Y_{b})+\bar{n} \delta(\bar{n})}} \nonumber\\
&\leq 2 \epsilon +    \sqrt{2^{r-\bar{n} \min_{b\in \mathcal{B}}H(X|Y_{b})+\bar{n} \delta(\bar{n})}}. \label{eqvar2}
\end{align}
From \eqref{eqvar2} and Lemma \ref{lemvarmi}, we conclude that the concealment requirement \eqref{eqconcealbr} holds with $r$ such that $ \lim_{n\to\infty} \frac{r}{n} \leq \min_{b\in \mathcal{B}} H(X|Y_{b}) - \delta$, $\delta>0$.

\subsubsection{Bindingness}
In the reveal phase, if the verifier observes $\tilde{x}^n$, then, by Lemma \ref{lem1}, a successful joint typicality test at the verifier requires $d_H(\tilde{x}^n, {x}^n) \leq {O(n^\alpha)}$, for some $\alpha<1$. 
This implies that Test~(ii) at the verifier in the reveal phase  can only succeed with a probability at most $2^{-n\eta}\cdot 2^{O(n^\alpha \log n)}$, which vanishes to zero as $n\to \infty$.

\subsection{Converse} \label{appB}
Again, because of the similarity with the converse proof of Theorem~\ref{th2}, we again only highlight the main steps of the proof.

We begin by proving Lemma \ref{lemmarkov2}, which is a counterpart to Lemma \ref{lemmarkov}. Although the proofs of these two lemmas are similar, we provide a detailed explanation here to clarify how the same technique can be applied in a multi-verifier setting, with some modifications based on the definitions in Section~\ref{sec:single-bidder-multi}.

\begin{lem} \label{lemmarkov2}
 Fix   $b \in \mathcal{B}$. Then, for  $(A,S)-(M,X^n)-(Y_{b}^n,S'_{b})$ forms a Markov chain.

\end{lem}
\begin{proof}

For   $i \in \llbracket 1 , n \rrbracket$, $j \in \llbracket 1 , r_i \rrbracket$, define $M_{1:i,1:j} \triangleq (M_{b,1:i,1:j})_{b\in\mathcal{B}}$, $M'_{1:i,1:j} \triangleq (M'_{b,1:i,1:j})_{b\in\mathcal{B}}$, $\bar{M}_{1:i,1:j} \triangleq (M_{1:i,1:j},M'_{1:i,1:j})$, $\bar{M}_{1:i} \triangleq \bar{M}_{1:i,1:r_i}$, $b\in\mathcal{B}$.
We have
\begin{align*}
&I(A S ;Y_b^n S'_b | \bar{M}_{1:n}X^n) \\
    & \leq I(A S ;Y_{\mathcal{B}}^n S'_{\mathcal{B}} | \bar{M}_{1:n}X^n) \numberthis \label{eqr12}\\
    & = I(AS;Y_{\mathcal{B}}^n S'_{\mathcal{B}} | \bar{M}_{1:n-1} \bar{M}_{n,1:r_n-1} M_{n,r_n}M'_{n,r_n} X^n) \\
    & \leq I(AS;Y_{\mathcal{B}}^n S'_{\mathcal{B}} M'_{n,r_n}| \bar{M}_{1:n-1} \bar{M}_{n,1:r_n-1} M_{n,r_n}  X^n) \\
    & \stackrel{(a)}= I(AS;Y_{\mathcal{B}}^n S'_{\mathcal{B}} | \bar{M}_{1:n-1} \bar{M}_{n,1:r_n-1} M_{n,r_n}  X^n) \\
    & \leq I(ASM_{n,r_n};Y_{\mathcal{B}}^n S'_{\mathcal{B}}| \bar{M}_{1:n-1} \bar{M}_{n,1:r_n-1} X^n) \\
     & \stackrel{(b)} = I(AS;Y_{\mathcal{B}}^n S'_{\mathcal{B}}| \bar{M}_{1:n-1} \bar{M}_{n,1:r_n-1} X^n) \numberthis \label{eqr22}\\
        & \stackrel{(c)}\leq I(AS;Y_{\mathcal{B}}^n S'_{\mathcal{B}} | \bar{M}_{1:n-1}  X^n) \\
    & \stackrel{(d)}= I(AS;Y_{\mathcal{B}}^{n-1} S'_{\mathcal{B}} | \bar{M}_{1:n-1}   X^n) \\
    & \leq I(AS X_n ;Y_{\mathcal{B}} ^{n-1} S'_{\mathcal{B}}  | \bar{M}_{1:n-1}   X^{n-1}) \\
    & \stackrel{(e)} = I(AS ;Y_{\mathcal{B}}^{n-1} S'_{\mathcal{B}} | \bar{M}_{1:n-1}   X^{n-1}) \numberthis \label{eqr32}\\
    & \stackrel{(f)}\leq I(AS ; S'_{\mathcal{B}}  )\\
    & = 0,
\end{align*}
where
\begin{enumerate}[(a)]
    \item holds because $M'_{n,r_n}$ is a function of $(S'_{\mathcal{B}}, M_{1:n,1:r_n},Y_{\mathcal{B}}^n)$;
    \item holds because $M_{n,r_n}$ is a function of $(A,S, M'_{1:n,1:r_n-1})$;
    \item  holds by repeating $r_n-1$ times the steps between \eqref{eqr12} and \eqref{eqr22};
    \item  holds because $(Y_{\mathcal{B}})_n - ( \bar{M}_{1:n-1},X^n, Y_{\mathcal{B}}^{n-1}, S'_{\mathcal{B}} )-(A,S)$;
    \item  holds because $X_n$ is a function of $(A,S, M'_{1:n-1,1:r_{n-1}})$;
    \item  holds by repeating $n-1$ times the steps between \eqref{eqr12} and~\eqref{eqr32}.
    \end{enumerate}
\end{proof}

Using Lemma \ref{lemmarkov2}, similar to Lemma \ref{lem}, one can obtain the following lemma.

\begin{lem} \label{lem8}
 Fix $b\in \mathcal{B}$. There exists $\hat{A}(V_{b},  X^n)$ such that $$ \lim_{n\to\infty}    \mathbb{P}[ \hat{A}(V_{b},  X^n) \neq {A}]  = 0.$$
\end{lem}

Finally, using Lemma \ref{lem8}, similar to \eqref{eqconvrate},  for any $b\in \mathcal{B}$, one can show that $n H(X |Y_{b}) \geq    nR  - o(n)$, 
which implies
\begin{align*}
 n \min_{b\in \mathcal{B}} H(X |Y_{b}) \geq    nR  - o(n).
\end{align*}
\section{Concluding Remarks}\label{seccl}
We have investigated multi-user commitment in two distinct settings. In the first setting, a verifier interacts with $L$ bidders, each committing to individual messages. This setting explores whether a multi-bidder protocol can outperform single-bidder protocols and time-sharing. Our results answer this question positively, and establish the sum-rate capacity when the bidders are colluding and the capacity region when they are non-colluding. In the second setting, a single bidder can interact with multiple verifiers. Our results characterize the commitment capacity, showing that positive commitment rates can still be achieved even if some verifiers drop out of the network after the commit phase, a scenario where commitment would otherwise be impossible with only a single verifier.

\bibliographystyle{IEEEtran}
\bibliography{bib}

\end{document}